\documentclass[a4paper,UKenglish,cleveref, autoref, thm-restate]{lipics-v2021}

\usepackage[procnumbered, linesnumbered,algoruled,noend ]{algorithm2e}
\usepackage{esvect}

\title{General Knapsack Problems in a Dynamic Setting} 

\author{Yaron Fairstein}{Computer Science Department, Technion, Haifa 3200003, Israel }{yyfairstein@gmail.com}{}{}
\author{Ariel Kulik}{Computer Science Department, Technion, Haifa 3200003, Israel }{naor@cs.technion.ac.il}{}{}
\author{Joseph (Seffi) Naor}{Computer Science Department, Technion, Haifa 3200003, Israel }{kulik@cs.technion.ac.il}{}{}
\author{Danny Raz}{Computer Science Department, Technion, Haifa 3200003, Israel }{danny@cs.technion.ac.il}{}{}

\authorrunning{Y. Fairstein, A. Kulik, J. Naor and D. Raz} 

\Copyright{Yaron Fairstein, Ariel Kulik, Joseph Naor and Danny Raz} 



\relatedversion{https://arxiv.org/abs/2105.00882} 

\EventEditors{Mary Wootters and Laura Sanit\`{a}}
\EventNoEds{2}
\EventLongTitle{Approximation, Randomization, and Combinatorial Optimization. Algorithms and Techniques (APPROX/RANDOM 2021)}
\EventShortTitle{\mbox{\scriptsize APPROX/RANDOM 2021}}
\EventAcronym{APPROX/RANDOM}
\EventYear{2021}
\EventDate{August 16--18, 2021}
\EventLocation{University of Washington, Seattle, Washington, US (Virtual Conference)}
\EventLogo{}
\SeriesVolume{207}
\ArticleNo{15}

\begin{document}
\nolinenumbers
\maketitle

\begin{abstract}
The world is dynamic and changes over time, thus any optimization problem used to model real life problems must address this dynamic nature, taking into account the cost of changes to a solution over time. 
The multistage model was introduced with this goal in mind. In this model we are given a series of instances of an optimization problem, corresponding to different times, and a solution is provided for each instance. The strive for obtaining near-optimal solutions for each instance on one hand, while maintaining similar solutions for consecutive time units on the other hand, is quantified and integrated into the objective function.
In this paper we consider the Generalized Multistage $d$-Knapsack problem, a generalization of the multistage variants of the Multiple Knapsack problem, as well as the $d$-Dimensional Knapsack problem. We present a PTAS for Generalized Multistage $d$-Knapsack. 
\end{abstract}
\newpage

\section{Introduction}\label{sec:intro}

In many optimization settings, the problem of interest is defined over a time horizon in which the actual setting evolves, resulting in changes over time to the problem constraints and the objective function.
Thus, even if the optimization problem at hand can be solved efficiently for a single time unit, it may not be clear how to extend this solution to a time-evolving setting.

An example of such a setting comes from the world of cloud management. A cloud provider maintains a data center with servers and offers clients virtual machines having different processing capabilities. Each client demands a virtual machine (with certain properties), and if provided it must pay for it. It would be na\"{\i}ve to assume that the demand of clients is static over time. Factors, such as peak vs. off-hours, and the day of the week, might affect client demand. Also, the cloud provider might either turn off servers to reduce hardware deterioration and electricity usage, or open more servers to meet higher demand.  Thus, the optimization problem is partitioned into multiple {\em stages}, where in each stage there are different constraints and possibly a different optimization goal. 

A simple solution is to ignore the dynamicity of the problem, and solve each stage separately and independently of other stages. Thus, profit at each stage is maximized, ignoring the solutions computed for the other stages. Such a solution may result in disgruntled clients, as it can lead to intermittent service between stages. Instead, we will aim for a {\em multistage} solution that balances between the optimum of each stage, while preserving some continuity between consecutive stages. This will be achieved by incorporating the continuity of the solution into the overall profit.

The multistage model was first introduced by Gupta et al. \cite{gupta2014changing} and Eisenstat et al. \cite{eisenstat2014facility} to address dynamic environments. Since its introduction, it has received growing attention (examples include \cite{an2017dynamic,fairstein2018algorithms,bampis2018multistage,bampis2019multistage,fluschnik2019multistage,deng2020approximation}). In the multistage model we are given a sequence of instances of an optimization problem. A solution constitutes of a series of solutions, one for each instance. 

Two different ideas were used to enforce a balance between single stage optimality and continuity. In \cite{eisenstat2014facility,gupta2014changing} a change cost is charged for the dissimilarity of consecutive solutions, while in \cite{bampis2019multistage} additional gains were given for their similarity. In the aforementioned cloud management problem, the change cost can be interpreted as installation costs and eviction costs charged when a client is initially served, and then its service is discontinued. The gains can be modeled as increased costs the client is charged to guarantee the continuity of its service.

The cloud management problem described above can be viewed as a multistage problem where the underlying optimization problem is the Multiple Knapsack problem (MKP). In MKP we are given a set of items, each associated with a weight and a profit. Also, we are given a set of bins, each one having a capacity. A feasible solution for MKP is an assignment of items to bins such that the total weight of the items assigned to each bin does not exceed its capacity. The objective is to find a feasible solution maximizing the profit accrued from the assigned items. In the context of the cloud management setting, the items are the virtual machine demands of the clients and the bins are the available servers.


\subsection{Problem Definition}\label{sec:definition}
We study the Generalized Multistage $d$-Knapsack problem. We begin with an informal description of the problem. An instance of the problem consists of $T$ stages, where in each stage we are given an instance of a generalization of the classic knapsack problem. While the instances differ between stages, in all stages the same set of items $I$ can  be packed. 
The continuity of the solution is enforced by quantifying the similarity of consecutive solutions and integrating it into the objective function.

We quantify continuity by four types of values. The first two values specify gains earned for the similarity of solutions. For example, if an item $i$ is packed in stages $t-1$ and $t$, gain $g^+_{i,t}$ is awarded. Similarly, $g^-_{i,t}$ is awarded if $i$ is not packed in $t-1$ and $t$. The other two values define the cost of changes between consecutive solutions. For example, if an item $i$ was not packed in stage $t-1$, and it is decided to pack it in stage $t$, a change cost of $c^+_{i,t}$ is charged. Similarly, $c^-_{i,t}$ is charged if $i$ is packed in $t$, but not in $t+1$.

The packing problem at each stage generalizes the Multiple Knapsack problem, as well as the $d$-Dimensional Knapsack problem. In each instance of the problem we are given $d$ sets of bins, and the weight an item occupies in a bin depends on the set to which the bin belongs to. The profit of an item is accrued once it is assigned to some bin in all $d$ sets of bins. This problem is called $d$-Multiple Knapsack Constraints Problem and is formally defined below.

A Multiple Knapsack Constraint (MKC) is a tuple $K=(w,B,W)$ defined over a set of items $I$. The function $w:I\rightarrow \mathbb{R}_+$ defines the weight of the items, $B$ is a set of bins, each equipped with a capacity defined by the function $W:B\rightarrow\mathbb{R}_+$. An {\em assignment} is a function $A:B\rightarrow 2^I$, defining which items are assigned to each of the bins. An assignment is feasible if $w(A(b))=\sum_{i\in A(b)}w_i \leq W(b)$ for each bin $b\in B$. Similarly, given a tuple of MKCs $\mathcal{K}=(K_j)_{j=1}^d$ over $I$, a tuple of $d$ assignments $\mathcal{A}=(A^j)_{j=1}^d$ is feasible for $\mathcal{K}$ if for each $j=1,\ldots,d$ assignment $A^j$ is a feasible assignment for $K_j$. We say $A$ is an assignment of set $S\subseteq I$ if $S=\cup_{b\in B}A_b$.

In $d$-Multiple Knapsack Constraints Problem ($d$-MKCP), a problem first introduced in \cite{fairstein2020tight}, we are given a tuple $\left( I, \mathcal{K},p \right)$, where $I$ is a set of items, $\mathcal{K}$ is a tuple of $d$ MKCs and $p:I\rightarrow \mathbb{R}_{\geq 0}$ defines the profit of each item. A feasible solution for $d$-MKCP is a set $S\subseteq I$ and a tuple of feasible assignments ${\mathcal{A}}$ (w.r.t $\mathcal{K}$) of $S$. The goal is to find a feasible solution that maximizes $p(S)=\sum_{i\in S}p(i)$. We note that if there exists an item with negative profit it can be discarded in advance. This fact is used later on, in Section \ref{sec:reduction}.

The Generalized Multistage $d$-Knapsack problem ($d$-GMK), is the multistage model of $d$-MKCP. The problem is defined over a time horizon of $T$ stages as follows. An instance of the problem is a tuple $\left((\mathcal{P}_t)_{t=1}^T, g^+,g^-, c^+,c^- \right)$, where  $\mathcal{P}_t = \left(I,\mathcal{K}_t,p_t\right)$ is a $d_t$-MKCP instance with $d_t\leq d$ for $t\in[T]$, $g^+,g^-\in \mathbb{R}_+^{I\times[2,T]}$ are the gain vectors and $c^+,c^-\in \mathbb{R}_+^{I\times[1,T]}$ are the change cost vectors.\footnote{We use the notations $[n,m]=\{ i\in \mathbb{N}~|~n\leq i\leq m\}$ and $[n]=[1,n]$ for $n,m\in \mathbb{N}$.}  We use $g^+_{i,t}$ and $g^-_{i,t}$ to denote the gain of item $i$ at stage $t$. Similarly, we use $c^+_{i,t}$ and $c^-_{i,t}$ to denote the change cost of item $i$ at stage $t$. 

A feasible solution for $d$-GMK is a tuple $(S_t,\mathcal{A}_t)_{t=1}^T$, where $(S_t,\mathcal{A}_t)$ is a feasible solution for $\mathcal{P}_t$ (note that $\mathcal{A}_t$ is a tuple of assignments of $S_t$). Throughout the paper we assume $S_0=S_{T+1}=\emptyset$ and denote the objective function of instance $\mathcal{Q}$ by $f_{\mathcal{Q}}:I^T\rightarrow \mathbb{R}$, where
\begin{align*}
    f_{\mathcal{Q}}\left((S_t)_{t=1}^T\right)=&\sum_{t=1}^T\sum_{i\in S_t} p_t(i) + \sum_{t=2}^{T}\left(\sum_{i\in S_{t-1}\cap S_{t}}g^+_{i,t} + \sum_{i\notin S_{t-1}\cup S_{t}}g^-_{i,t}\right) \\
    &~~~~~- \sum_{t=1}^{T}\left( \sum_{i\in S_{t}\setminus S_{t-1}} c^+_{i,t} + \sum_{i\in S_{t}\setminus S_{t+1}} c^-_{i,t} \right).
\end{align*}

The goal is to find a feasible solution that maximizes the objective function $f_\mathcal{Q}$.

A study of $d$-GMK reveals it does not admit a constant factor approximation algorithm (see Section~\ref{sec:d-dim-kp-hardness}). We found that in hard instances the change costs are much larger than the profits. Thus we consider an important parameter of the problem, the {\em profit-cost ratio}. It is defined as the maximum ratio, over all items, between the change cost ($c^+,c^-$) and the profit of an item over all stages. It is denoted by $\phi_\mathcal{Q}$ for any instance $\mathcal{Q}$, and is formally defined as 
\[ \phi_{\mathcal{Q}} = \min\left(\Big\{\infty\Big\}\bigcup \left\{r\geq 0 ~\Big|~ \forall i\in I, t_1,t_2\in[T] : ~\max\left\{c^+_{i,t_1},c^-_{i,t_1}\right\}\leq r\cdot p_{t_2}(i)\right\}\right) \]
We show that $d$-GMK instances where the profit-cost ratio is bounded by a constant admit a PTAS.

We also consider Subdmodular $d$-GMK, a submodular variant of $d$-GMK where the profit functions are replaced with monotone submodular set functions. A set function $p:2^I\rightarrow \mathbb{R}$ is submodular if for every $A\subseteq B\subseteq I$ and $i\in I\setminus B$ it holds that $p(A\cup\{i\})-p(A)\geq p(B\cup\{i\})-p(B)$. Submodular functions appear naturally in many settings such as coverage \cite{feige1998threshold}, matroid rank \cite{calinescu2007maximizing} and cut functions \cite{feige1995approximating}. We use similar techniques to develop the algorithms for $d$-GMK and Submodular $d$-GMK. Thus, we focus on $d$-GMK and  defer the formal definition as well as the algorithm for Submoduar $d$-GMK to Appendix \ref{app:submodular}.

Both $d$-GMK and Submodular $d$-GMK generalize the Multistage Knapsack problem recently considered by Bampis et al. \cite{bampis2019multistage}. There are several aspects by which it is generalized. First, handling multiple knapsack constraints as well as $d$-dimensional knapsack vs a single knapsack in \cite{bampis2019multistage}. Second, the profit earned from assigning items can be described as a submodular function, not only by a modular function. Third, \cite{bampis2019multistage} considered only symmetric gains, i.e., the same gain is earned whether an item is assigned or not assigned in consecutive stages. Lastly, change costs were not considered in \cite{bampis2019multistage}.

\subsection{Our Results}\label{sec:results}
Our main result is stated in the following theorem.
\begin{theorem}\label{thm:ptas-for-modular}
For any fixed $d\in\mathbb{N}$ and $\phi\geq 1$ there exists a randomized PTAS for $d$-GMK with a profit-cost ratio bounded by  $\phi$.
\end{theorem}
The result uses the general framework of~\cite{bampis2019multistage}, in which the authors first presented an algorithm for instances with bounded time horizon, and then showed how it can be scaled for general instances.
To handle bounded time horizons we show an approximation factor preserving reduction (as defined in \cite{vazirani2013approximation}\footnote{A formal definition is provided in Appendix \ref{app:omitted-proofs} for completeness}) from $d$-GMK to a generalization of $q$-MKCP.
The reduction illuminates the relationship between $d$-GMK and $q$-MKCP.  As $q$-MKCP admits a PTAS \cite{fairstein2020tight}, this results in a PTAS for  $d$-GMK instances with a bounded time horizon.

We note the reduction can be applied to the problem considered in~\cite{bampis2019multistage} as well. In this case the target optimization problem is $d$-dimensional knapsack with a matroid constraint. As the latter problem is known to admit a PTAS \cite{grandoni2010approximation}, this suggest a simpler solution for bounded time horizon in comparison to the one given in \cite{bampis2019multistage}.

 To generalize the result to unbounded time horizon we use an approach similar to \cite{bampis2019multistage}, though a more sophisticated analysis was required to handle the change costs. The generalization is  achieved by cutting the time horizon into sub-instances with a fixed time horizon. Each sub-instance is solved separately, and then the solutions are combined to create a solution for the full instance. 
 Handling change costs is trickier as cutting an instance may lead to an excessive charge of change costs at the cut points. We must compensate for these additional costs, or we will not be able to bound the value of the solution.

The results for the modular variant generalizes the PTAS  for Multistage Knapsack \cite{bampis2019multistage}. For $d\geq 2$, we cannot expect better results as even $d$-KP, also generalized by $d$-GMK, does not admit an {\em efficient} PTAS (EPTAS). Theorem~\ref{thm:no-eptas} shows an EPTAS cannot be obtained for $1$-GMK  as well
\begin{theorem}\label{thm:no-eptas}
Unless $W[1]=FPT$, there is no EPTAS for $1$-GMK, even if the length of the time horizon is $T=2$, the set of bins in each MKC contains one bin and there are no change costs.
\end{theorem}
The theorem is proved using a simple reduction from $2$-dimensional knapsack.
Bampis et al. \cite{bampis2019multistage} considered a similar withered down instance and proved that even if the gains are symmetric (i.e., $g^+_{i,t}=g^-_{i,t}$) a {\em Fully} PTAS (FPTAS) does not exist for the problem. 

Using a reduction from multidimensional knapsack we show that $1$-GMK, in its general form, cannot be approximated to any constant factor. 
\begin{theorem}\label{thm:phi-hardness}
For any $d\geq 1$, 
there is no polynomial time  approximation algorithm for $d$-GMK with a constant approximation ratio, unless $NP=ZPP$.
\end{theorem}
This result justifies our study of the special cases of $d$-GMK in which the profit-cost ratio is bounded by a constant.

The techniques used to develop the algorithm for $d$-GMK can be adjusted slightly to produce an approximation algorithm for Submodular $d$-GMK.
\begin{theorem}\label{thm:submodular-result}
For any fixed $d\in\mathbb{N}$ and $\epsilon>0$ there exists a randomized $\left(1-\frac{1}{e}-\epsilon\right)$-approximation algorithm for Submodular $d$-GMK.
\end{theorem}

In the submodular variant one cannot hope for vast improvement over our results as the algorithm is almost tight. This is due to the hardness results for submodular maximization subject to a cardinality constraint presented by Nemhauser and Wolsey \cite{nemhauser1978best}.

\subsection{Related Work}\label{sec:related}
In the multistage model we are given a series of instances of an optimization problem, and we search for a solution which optimizes each instance while maintains some similarity between solutions. Many optimization problem were considered under this framework. These include matching \cite{bampis2018multistage,chimani2020approximating}, clustering \cite{deng2020approximation}, subset sum \cite{bampis2019online}, vertex cover \cite{fluschnik2019multistage} and minimum $s-t$ path \cite{fluschnik2020multistage}.

The multistage model was first presented by both Eisenstat et al. \cite{eisenstat2014facility} and Gupta et al. \cite{gupta2014changing}. In \cite{eisenstat2014facility} the Multistage Facility Location problem was considered, where the underlying metric in which clients and facility reside changes over time. A logarithmic approximation algorithm was presented for two variants of the problem; the hourly opening costs where the opening cost of a facility is charged at each stage in which it is open, and the fixed opening costs where a facility is open at all stages after its opening costs is paid. A logarithmic hardness result was also presented for the fixed opening costs variant. An et al. \cite{an2017dynamic} improved the result for the hourly opening costs variant and presented a constant factor approximation. Fairstein et al. \cite{fairstein2018algorithms} proved that the logarithmic hardness result does not hold if only the client locations change over time and the facilities are static.

As mentioned, the multistage model was also introduced by Gupta et al. \cite{gupta2014changing}, where the Multistage Matroid Maintenance (MMM) problem was considered. In MMM we are given a set of elements equipped with costs that change over time. In addition, we are given a matroid. The goal is to select a base of minimum costs at each stage whilst minimizing the cost of the difference of bases selected for consecutive stages. Gupta et al. \cite{gupta2014changing} presented a logarithmic approximation algorithm for the problem, as well as fitting lower bound, proving this result is tight.

\paragraph*{Organization.} 
Section~\ref{sec:PTAS} provides the approximation schemes from Theorem~\ref{thm:ptas-for-modular}.
Hardness results are presented in Section \ref{sec:hardness}. The extension to Submodular $d$-GMK is given in Appendix~\ref{app:submodular}.

\section{Approximation Scheme for \texorpdfstring{$d$} --GMK}
\label{sec:PTAS}
In this section we derive the approximation scheme for $d$-GMK with bounded profit-cost ratio. 
In Section~\ref{sec:reduction} we show how a PTAS for instances with bounded time horizon can be 
obtain via a reduction to a variant of $d$-MKCP. Subsequently, in Section~\ref{sec:breaking} we show how the algorithm for bounded time horizon can be used to approximate general instance. 
\subsection{Bounded Time Horizon}\label{sec:reduction}
In this section we provide a reduction from an instance of $q$-GMK to a generalization of $d$-MKCP (for specific values of $d$ and $q$). The generalization was presented in \cite{fairstein2020tight} and is called $d$-MKCP With A Matroid Constraint ($d$-MKCP+) and is defined by a tuple $(I,\mathcal{K},p,\mathcal{I})$, where $(I,\mathcal{K},p)$ forms an instance of $d$-MKCP. Also, the set $\mathcal{I}\subseteq 2^I$ defines a matroid\footnote{A formal definition for matroid can be found in~\cite{schrijver2003combinatorial}} constraint. A feasible solution for $d$-MKCP+ is a set $S\in\mathcal{I}$ and a tuple of feasible assignments ${\mathcal{A}}$ (w.r.t $\mathcal{K}$) of $S$. The goal is to find a feasible solution which maximizes $\sum_{i\in S}p(i)$. The following definition presents the construction of the reduction.

\begin{definition}\label{def:modular-reduction}
Let $\mathcal{Q}=\left((\mathcal{P}_t)_{t=1}^T,g^+,g^-,c^+,c^-\right)$ be an instance of $d$-GMK, where $\mathcal{P}_{t} = \left(I,\mathcal{K}_{t},p_t\right)$ and $\mathcal{K}_t = (K_{t,j})_{j=1}^{d_t}$. Define $R(\mathcal{Q})=\left(E,\tilde{\mathcal{K}},\tilde{p},\mathcal{I}\right)$ where
\begin{itemize}
    \item $E=I\times 2^{[T]}$
    \item $\mathcal{I}=\left\{ S\subseteq E~\middle|~ \forall i\in I :~\left|S\cap \left(\{i\}\times2^{[T]}\right)\right|\leq 1 \right\}$
    \item For $t\in [T], j\leq d_t$ set MKC $\tilde{K}_{t,j}=(\tilde{w}_{t,j},B_{t,j},W_{t,j})$ over $E$, \\where $K_{t,j} = (w_{t,j},B_{t,j},W_{t,j})$ and
    \[ \tilde{w}_{t,j}((i,D)) = 
\begin{cases}
	w_{t,j}(i) & t\in D\\
	0 & \textnormal{otherwise}
\end{cases} \]
    \item For $t=1,...,T, d_t<j\leq d$ set MKC $\tilde{K}_{t,j}= (w_0,\{b\},W_0)$ over $E$, where $w_0:2^E\rightarrow \{0\}$, $W_0(b)=0$ and $b$ is an arbitrary bin (object).
    \item $\tilde{\mathcal{K}} = \left( \tilde{K}_{t,j} \right)_{t\in[T],j\in [d]}$.
    \item The objective function $\tilde{p}$ is defined as follows.
    \[ \tilde{p}(S) = \sum_{(i,D)\in S}\left(\sum_{t\in D}p_t(i)+\sum_{t\in D : t-1\in D}g^+_{i,t} +\sum_{t\notin D : t-1\notin D}g^-_{i,t} - \sum_{t\in D: t-1\notin D}c^+_{i,t} - \sum_{t\in D: t+1\notin D}c^-_{i,t} \right) \]
\end{itemize}
\end{definition}

Each element $(i,D)\in E$ states the subset of stages in which item $i$ is assigned. I.e., $i$ is only assigned in stages $t\in D$. Thus any solution should include at most one element $(i,D)$ for each $i\in I$. This constraint is fully captured by the partition matroid constraint defined by the set of independent sets $\mathcal{I}$. Finally, if an element $(i,D)$ is selected, we must assign $i$ in each MKC $K_{t,j}$ for $j\in [d_t], t\in D$. This is captured by the weight function $\tilde{w}$, as an element $(i,D)$ weighs $w_{t,j}(i)$ if and only if $t\in D$ (otherwise its weight is zero and it can be assigned for ``free'').

\begin{lemma}
For any $d$-GMK instance $\mathcal{Q}$ with time horizon $T$, it holds that  
$R(\mathcal{Q})$ is a $dT$-MKCP+ instance.
\end{lemma}

\begin{proof}
Let $\mathcal{Q}=\left((\mathcal{P}_t)_{t=1}^T,g^+,g^-,c^+,c^-\right)$ be an instance of $d$-GMK, where $\mathcal{P}_{t} = \left(I,\mathcal{K}_{t},p_t\right)$ and $\mathcal{K}_t = (K_{t,j})_{j=1}^{d_t}$. Also, let $R(\mathcal{Q})=\left(E,\tilde{\mathcal{K}},\tilde{p},\mathcal{I}\right)$ be the reduced instance of $\mathcal{Q}$ as defined in Definition \ref{def:modular-reduction}. It is easy to see that the set $\mathcal{I}$ is the independent sets of a partition matroid, as for each item $i$ at most one element $(i,D)$ can be chosen. Thus, $\mathcal{I}$ is the family of independent sets of a matroid as required. 

Next, $\mathcal{K}$ defines a tuple of MKCs, so all that is left to prove is that $\tilde{p}$ is non-negative and modular. For each element $(i,D)\in E$ we can define a fixed value
\[ v((i,D)) = \sum_{t\in D}p_t(i) + \sum_{t\in D :~ t-1\in D}g^+_{i,t} +\sum_{t\notin D :~ t-1\notin D}g^-_{i,t} - \sum_{t\in D :~ t-1\notin D}c^+_{i,t} +\sum_{t\in D :~ t+1\notin D}c^-_{i,t} \]
It immediately follows that $\tilde{p}(S) = \sum_{e\in S}v(e)$ and that $\tilde{p}$ is modular.
As stated in Section \ref{sec:definition}, elements with negative values are discarded in advance such that $\tilde{p}$ is also non-negative.
\end{proof}

\begin{lemma}\label{lem:modular-reduction-dir1}
Let $\mathcal{Q}$ be an instance of $d$-GMK with time horizon $T$. For any feasible solution $(S_t,\mathcal{A}_t)_{t=1}^T$ of $\mathcal{Q}$ there exists a feasible solution $\left(S,(\tilde{A}_{t,j})_{t\in [T],j\in [d]}\right)$ of $R(\mathcal{Q})$  such that $f_\mathcal{Q}\left((S_t)_{t=1}^T\right) = \tilde{p}(S)$. 
\end{lemma}

\begin{proof}
Let $\mathcal{Q}=\left((\mathcal{P}_t)_{t=1}^T,g^+,g^-,c^+,c^-\right)$ be an instance of $d$-GMK, where $\mathcal{P}_{t} = \left(I,\mathcal{K}_{t},p_t\right)$ and $\mathcal{K}_t = (K_{t,j})_{j=1}^{d_t}$. Also, let $R(\mathcal{Q})=\left(E,\tilde{\mathcal{K}},\tilde{p},\mathcal{I}\right)$ be the reduced instance of $\mathcal{Q}$, where $\tilde{\mathcal{K}}=\left(\tilde{K}_{t,j }\right)_{t\in [T], j\in [d]}$
and $\tilde{K}_{t,j}=(\tilde{w},B_{t,j}, W_{t,j})$ (see Definition \ref{def:modular-reduction}). Consider some feasible solution $(S_t,\mathcal{A}_t)_{t=1}^T$ for $\mathcal{Q}$, where $\mathcal{A}_t = (A_{t,j})_{j=1}^{d_t}$. In the following we define a solution $\left(S, \left(\tilde{A}_{t,j}\right)_{t\in [T],~j\in [d]}\right)$ for  $R(\mathcal{Q})$.
Let 
\[S= \left\{(i,D)~\middle|~i\in I, ~D=\{t\in[T]~|~i\in S_t \}\right\}\]
It can be easily verified that $S\in \mathcal{I}$. The value of the subset $S$ is 
\[
\begin{aligned}
    &\tilde{p}(S) = \\
    &\sum_{(i,D)\in S}\Bigg(\sum_{t\in D}p_t(i)+\sum_{t\in D :~ t-1\in D} g^+_{i,t} +\sum_{t\notin D :~ t-1\notin D}g^-_{i,t} 
    - \sum_{t\in D:~ t-1\notin D}c^+_{i,t} - \sum_{t\in D:~ t+1\notin D}c^-_{i,t} \Bigg)= \\
   &  \sum_{t=1}^T\sum_{i\in S_t} p_t(i) + \sum_{t=2}^{T}\left( \sum_{i\in S_{t-1}\cap S_{t}} g^+_{i,t} + \sum_{i\notin S_{t-1}\cup S_{t}} g^-_{i,t} \right)
    - \sum_{t=1}^{T}\left( \sum_{i\in S_{t}\setminus S_{t-1}} c^+_{i,t} + \sum_{i\in S_t\setminus S_{t+1}} c^-_{i,t} \right) = \\
    & f_{\mathcal{Q}}\left((S_t)_{t=1}^T \right).
    \end{aligned}
\]
Next, for each $t\in[T],j\in[d]$ we present an assignment $\tilde{A}_{t,j}$ of $S$. Consider the following two cases:
\begin{enumerate}
    \item If $j>d_t$, recall $\tilde{K}_{t,j}=(w_0, \{b\} , W_0)$ where $w_0(i,D)=0$ for all $(i,D)\in E$ and $W_0(b)=0$. We define $\tilde{A}_{t,j}$ by $\tilde{A}_{t,j}(b)=S$. It thus holds that $w_0(\tilde{A}_{t,j}(b))=0=W_0$. That is, $\tilde{A}_{t,j}$ is feasible. 
    \item  If $j\leq d_t$, let $b^*\in B_{t,j}$ be some unique bin in $B_{t,j}$ and define assignment $\tilde{A}_{t,j}:B_{t,j}\rightarrow2^E$ by 
    \begin{equation}
    \label{eq:red_assign}
    \begin{aligned}
        &\tilde{A}_{t,j}(b) = \left(A_{t,j}(b)\times 2^{[T]}\right)\cap S &\forall b\in B_{t,j}\setminus\{b^*\}\\
        &\tilde{A}_{t,j}(b^*) = \left( \left(A_{t,j}(b^*)\times 2^{[T]}\right)\cap S \right) \cup \{(i,D)\in S~|~t\not\in D\}
    \end{aligned}
    \end{equation}
    The assignment $\tilde{A}_{t,j}$ is a feasible assignment w.r.t $\tilde{K}_{t,j}$ since for each bin $b\in B_{t,j}$ it holds that
\[ \sum_{(i,D)\in \tilde{A}_{t,j}(b)}\tilde{w}_{t,j}((i,D)) = \sum_{i\in A_{t,j}(b)}w_{t,j}(i) \leq W_{t,j}(b) \]

Let $(i,D)\in S$. If $i\in S_t$ there is $b\in B_{t,j}$ such that $i\in A_{t,j}(b)$, hence $(i,D)\in \tilde{A}_{t,j}(b)$ by \eqref{eq:red_assign}. If $i\not\in S_t$ then $t\notin D$ and thus $(i,D)\in \tilde{A}_{t,j}(b^*)$. Overall, we have $S\subseteq \bigcup_{b\in B_{t,j}} \tilde{A}_{t,j}(b)$. By \eqref{eq:red_assign} it follows that $S\supseteq \bigcup_{b\in B_{t,j}} \tilde{A}_{t,j}(b)$ as well, thus $S= \bigcup_{b\in B_{t,j}} \tilde{A}_{t,j}(b)$. I.e, $\tilde{A}_{t,j}$ is an assignment of $S$.
\end{enumerate}
Note that the assignments can be constructed in polynomial time. We can conclude that $\left(S,\left(\tilde{A}_{t,j}\right)_{t\in[A], j\in [d]}\right)$ is a feasible solution for $R(\mathcal{Q})$, and its value is $f_{\mathcal{Q}}\left((S_t)_{t=1}^T \right)$.
\end{proof}

\begin{lemma}\label{lem:modular-reduction-dir2}
Let $\mathcal{Q}$ be an instance of $d$-GMK (with arbitrary time horizon $T$). For any feasible solution $\left(S,(\tilde{A}_{t,j})_{t\in [T],j\in [d]}\right)$ for $R(\mathcal{Q})$ a feasible solution  $(S_t,\mathcal{A}_t)_{t=1}^T$ for $\mathcal{Q}$ such that $f_\mathcal{Q}\left((S_t)_{t=1}^T\right) = \tilde{f}(S)$ can be constructed in polynomial time.
\end{lemma}

The proof of Lemma \ref{lem:modular-reduction-dir2} is similar to the proof of Lemma \ref{lem:modular-reduction-dir1}, thus it is deferred to Appendix~\ref{app:omitted-proofs}.

For any $d$-GMK instance with a fixed time horizon $T$, the reduction $R(\mathcal{Q})$ can be constructed in polynomial (as $|E|=|I|\cdot 2^{|T|}$). The next corollary follows from this observation and lemmas \ref{lem:modular-reduction-dir1} and \ref{lem:modular-reduction-dir2}.

\begin{corollary}
For any fixed $T\in\mathbb{N}$, there exists an approximation factor preserving reduction from $d$-GMK with a time horizon bounded by $T$ to $dT$-MKCP+.
\end{corollary}

In \cite{fairstein2020tight} a PTAS for $d$-MKCP+ is presented. Thus, the next lemma follows from the above corollary.
\begin{lemma}\label{lem:ptas-for-fixed-horizon}
For any fixed $T\in\mathbb{N}$ there exists a randomized PTAS for $d$-GMK with a time horizon bounded by $T$.
\end{lemma}

\subsection{General Time Horizon}\label{sec:breaking}
In this section we present an algorithm for $d$-GMK with a general time horizon $T$. This is done by cutting the time horizon at several stages into sub-instances. Each sub-instance is optimized independently and then the solutions are combined to create a solution for the complete instance. A somewhat similar technique was used in \cite{bampis2019multistage}.  However, they considered a model without change costs which is much simpler. Our analysis is more delicate as it requires local consideration of the assignment of each item to ensure that any additional costs charged are covered by profit and gains earned. In Appendix \ref{app:submodular-cutting} we present the analysis for Submodular $d$-GMK. Since for this case there are no change costs it is more comparable to the analysis of \cite{bampis2019multistage}.  

Given an instance $\mathcal{Q} = \left((\mathcal{P}_t)_{t=1}^T,g^+,g^-,c^+,c^-\right)$ we  define a sub-instance for the sub-range $[t_1,t_2]$, denoted throughout this section by $\left((\mathcal{P}_t)_{t=t_1}^{t_2},g^+,g^-,c^+,c^-\right)$ without shifting or truncating the gain and change costs vectors. For example, for $t\in[t_1,t_2]$ gain $g^+_{i,t}$ is earned for assigning item $i$ in stages $t$ and $t+1$. Also, observe that stages $t_1-1$ and $t_2+1$ are outside the scope of the instance. Thus, when evaluating a solution for the sub-instance it is assumed that $S_{t_1-1}=S_{t_2+1}=\emptyset$.

Given an integer $T\in \mathbb{N}$, a set of {\it cut points} $U = \{u_0,...,u_k\}$ of $T$ is a set of integers such that for every $j=0,...,k-1$ it holds that $u_j < u_{j+1}$ and $1= u_0 < u_k = T+1$. 

\begin{definition}
Let $\mathcal{Q} = \left((\mathcal{P}_t)_{t=1}^T,g^+,g^-,c^+,c^-\right)$ be an instance of $d$-GMK, where $\mathcal{P}_{t} = \left(I,\mathcal{K}_{t},p_t\right)$. Also, let $U=\{u_0,\ldots,u_k\}$ be a set of cut points. The tuple of $d$-GMK instances $\mathcal{Q}_U=\left((\mathcal{P}_t)_{t=u_j}^{u_{j+1}-1},g^+,g^-,c^+,c^-\right)_{j=0}^{k-1}$ is defined as the cut instances of $\mathcal{Q}$ w.r.t $U$.
\end{definition}

\begin{definition}
Let $\mathcal{Q}$ be an instance of $d$-GMK, $U=\{u_0,...,u_k\}$ be a set of cut points and $\mathcal{Q}_U$ be the respective cut instances. Also, let $\left(\left(S_t,\mathcal{A}_t\right)_{t=u_j}^{u_{j+1}-1}\right)_{j=0}^{k-1}$ be a tuple of feasible solutions for the tuple of cut instances $\mathcal{Q}_U$. Then, the solution $\left(S_t,\mathcal{A}_t\right)_{t=1}^{T}$ for $\mathcal{Q}$ is called a {\it cut solution}. 
\end{definition}

The next corollary elaborates on the relationship between a cut solution and the cut instance solutions from which it is constructed.
\begin{corollary}\label{cor:cut-solution}
Given any $d$-GMK instance $\mathcal{Q}$, cut points $U$, cut instances $\mathcal{Q}_U$ and feasible solutions for the cut instances, the respective cut solution is a feasible solution for $\mathcal{Q}$, and its value is at least the sum of values of the solutions for the cut instances.
\end{corollary}
The proof of the corollary is fairly simple, and is defered to Appendix \ref{app:omitted-proofs}. We are now ready to present the algorithm for $d$-GMK with general time horizon length.

\begin{algorithm}[H]
	\SetAlgoLined
	\SetKwInOut{Input}{Input}
	\DontPrintSemicolon
	\Input{$0<\epsilon<\frac{1}{4},\phi\geq 1$, a $d$-GMK instance $\mathcal{Q}$ with time horizon $T$ such that $\phi_\mathcal{Q}\leq\phi$, and $\alpha$-approximation algorithm $A$ for $d$-GMK with time horizon $T\leq\frac{2\phi}{\epsilon^2}$.}
	
	Set $\mu = \frac{\epsilon^2}{\phi}$.\;
	
	\For{$j=1,...,\frac{1}{\mu}$}{
	    Set $U_j = \left\{\frac{a}{\mu}+j-1 ~\middle|~  a\in\mathbb{N},~a\geq 1,~ \frac{a}{\mu}+j-1\leq T-\frac{1}{\mu}\right\}\cup\{1,T+1\}$. \label{line:cut-points-set}\;
	    
	    Find a solution for each cut instance in $\mathcal{Q}_{U_j}$ using algorithm $A$ and set $\mathcal{S}_j$ as the respective cut solution.\;
	}
	
	Return the solution $\mathcal{S}_j$ which maximizes the objective function $f_\mathcal{Q}$.
	
	\caption{General Time Horizon}
	\label{alg:break-time-horizon}
\end{algorithm}

Before analysing the algorithm we present several definitions and lemmas that are essential for the proof. First, we start by reformulating the solution. Instead of describing the assignment of items by the tuple $(S_t)_{t=1}^T$, we define a new set of elements $E=I\times[T]\times [T]$, where each element $(i,t_1,t_2)\in E$ states that item $i$ is assigned in the interval $[t_1,t_2]$. Given a feasible solution $(S_t,\mathcal{A}_t)_{t=1}^T$ for $\mathcal{Q}$ we denote the representation of $(S_t)_{t=1}^T$ as a subset of $E$ by $E\left((S_t)_{t=1}^T\right)$ and it is equal to
\[ E\left((S_t)_{t=1}^T\right) = \left\{\left(i,t_1,t_2\right)\in E ~|~ \forall t\in [t_1,t_2] : i\in S_t ~\text{ and }~ i\notin S_{t_1-1}\cup S_{t_2+1}\right\}. \]
If $\tilde{S} = E\left((S_t)_{t=1}^T\right)$, we define the reverse mapping as $\tilde{S}(t) = \{i\in I ~|~ \exists (i,t_1,t_2)\in S : t\in [t_1,t_2]\}= S_t$. Now, we can define a solution for $d$-GMK using our new representation as $(\tilde{S},\mathcal{A}_t)_{t=1}^T$.

\begin{definition}\label{def:element-value}
The {\em value}, $v(e)$, of element $e=(i,t_1,t_2)$ is defined as the total value earned from assigning $i$ in the range $[t_1,t_2]$ minus the change costs charge for assigning and discarding it. Formally,
\[v(e) = \sum_{t=t_1}^{t_2}p_t(i) + \sum_{t=t_1+1}^{t_2}g^+_{i,t} - c^+_{i,t_1} - c^-_{i,t_2} \]
The value of solution $\tilde{S}\subseteq E$ for $d$-GMK instance $\mathcal{Q}$ is $\sum_{e\in \tilde{S}}v(e) + \sum_{t=2}^T\sum_{i\notin \tilde{S}(t-1)\cup\tilde{S}(t)}g^-_{i,t} $ and it is equal to $f_{\mathcal{Q}}\left((\tilde{S}(t))_{t=1}^T\right)$.
\end{definition}

In Algorithm \ref{alg:break-time-horizon} we consider a solution for a tuple of cut instances created by cutting an instance at a set of cut points. Here we consider the opposite action, the effect of cutting a solution at these cut points. We start by considering a single cut point.

\begin{definition}
Given an element $e=(i,t_1,t_2)$ and a cut point $u\in(t_1,t_2]$ we define the outcome of cutting $e$ at $u$ as the set of intervals $u(e) = \{(i,t_1,u-1),(i,u,t_2)\}$. Also, the {\em loss} caused by cutting $e$ at $u$ is defined as the difference between the value of $e$ and the sum of value of elements in $u(e)$. It is denoted by $\ell(e,u)$ and is equal to
\begin{equation}\label{eq:loss-at-cutpoint}
\begin{aligned}
    \ell(e,u)=v(e)-\sum_{e'\in u(e)}v(e') = g^+_{i,u}+c^+_{i,u}+c^-_{i,u-1}
\end{aligned}
\end{equation}
\end{definition}

We can similarly extend the definition to include more than one cut point as follows.

\begin{definition}
Given a set of cut points $U=\{u_0,u_1,\ldots, u_k\}$ and an element $e=(i,t_1,t_2)$ define $U(e)$ as the set of elements created by cutting the $e$ at all cut points in $U$. Formally,
\[
U\left((i,t_1,t_2)\right) = \left\{\left(i,\max\{t_1,u_{r-1}\},\min\{u_r-1, t_2 \} \right) \Big|  r\in [k] \textnormal{ and } [u_{r-1},u_r -1] \cap  [t_1, t_2] \neq \emptyset \right\}
\]
\end{definition}
Consider the following example as a demonstration of the above definition. If $e=(i,t_1,t_2)$ and $(t_1,t_2]\cap U_j = \{u_2,u_3\}$, then $U_j(e) = \{(i,t_1,u_2-1),(i,u_2,u_3-1),(i,u_3,t_2)\}$. 

The definition of loss caused by cutting an element can be extended to a set of cut points $U$. If element $e=(i,t_1,t_2)$ is cut by a of cut points $U$ the loss is
\begin{equation}\label{eq:loss-by-set-of-cuts}
    \ell(e,U)=v(e)-\sum_{e'\in U(e)}v(e')=\sum_{u\in U\cap(t_1,t_2]}\left(g^+_{i,u}+c^+_{i,u}+c^-_{i,u-1}\right) =  \sum_{u\in U\cap(t_1,t_2]}\ell(e,u)
\end{equation} 
since only gains $g^+$ are lost due to cutting as well as change cost for splitting an assignment into two intervals. This means that even if an element is cut multiple times, the loss due to each cut point can be considered separately.

\begin{lemma}\label{lem:cutting-algorithm}
Let $0<\epsilon<\frac{1}{4}$, $\phi\geq 1$ and $A$ be an $\alpha$-approximation algorithm for $d$-GMK with time horizon $T\leq\frac{2\phi}{\epsilon^2}$. Also, let $\mathcal{Q}$ be an instance of $d$-GMK such that $\phi_{\mathcal{Q}}\leq \phi$. Algorithm \ref{alg:break-time-horizon} approximates $\mathcal{Q}$ within a factor of $(1-\epsilon)\alpha$.
\end{lemma}

\begin{proof}
Let $\mathcal{Q}=\left((\mathcal{P}_t)_{t=1}^T,g^+,g^-,c^+,c^-\right)$ be an instance of $d$-GMK with time horizon $T$ and profit-cost ratio $\phi_\mathcal{Q}\leq\phi$. We assume for simplicity $\phi$ is integral. Let $0<\epsilon<\frac{1}{4}$ and $\mu = \frac{\epsilon^2}{\phi}$. Also, let $A$ be an $\alpha$-approximation algorithm for $d$-GMK with time horizon $T'\leq\frac{2\phi}{\epsilon^2}=\frac{2}{\mu}$. Note, if $T\leq\frac{2}{\mu}$, the cut points set $U_0$ is an empty set, and in this case  $A$  returns an $\alpha$-approximation solution for $\mathcal{Q}$ as required. 

Let $U_j = \{u^j_1,\ldots, u^j_{k_j}\}$ for $j=1,\ldots \frac{1}{\mu}$. We show that there exists a set of cut points $U_j$ and a tuple of solutions $\left((S_t,\mathcal{A}_t)_{t=u^j_r}^{u^j_{r+1}-1}\right)_{r=0}^{k_j-1}$ for each cut instance in $\mathcal{Q}_{U_j}=(q_j^r)_{r=0}^{k_j-1}$, such that the sum of values of the solutions, $\sum_{r=0}^{k_j-1}f_{q_j^r}\left((S_t)_{t=u_r}^{u_{r+1}-1}\right)$, is sufficiently large. From Corollary \ref{cor:cut-solution} it follows that the value of a cut solution is larger than the sum of its parts (due to lost gains and change costs saved if an item is assigned in adjacent instances). Thus, this also proves that the maximum cut solution found is sufficiently large as well. 

Let $(S^*_t,\mathcal{A}^*_t)_{t=1}^T$ be an optimal solution for $\mathcal{Q}$, and let $\tilde{S}^* = E\left((S^*_t)_{t=1}^T\right)$. We partition $\tilde{S}^*$ into two subsets by the length of the interval they describe. Formally, $X=\left\{(i,t_1,t_2)\in \tilde{S}^* ~|~ t_2-t_1< \frac{\phi}{\epsilon}\right\}$ and $Y=\tilde{S}^*\setminus X$. So $X$ contains short intervals, and $Y$ contains long intervals.

Define $\tilde{S}_j$ as the subset of elements longer than $\phi$ in $\cup_{e\in Y}U_j(e)$ as well as short elements $e\in X$ that are not cut by $U_j$. I.e., 
\[ \tilde{S}_j = \{e\in X ~|~ U_j(e) = \{e\}\}\cup\bigcup_{e\in Y}\left\{(i,t_1,t_2)\in U_j(e) ~|~ t_2-t_1\geq \phi\right\} \]
At each stage $t\in [T]$ it holds that $\tilde{S}_j(t)\subseteq S^*_t$. Thus there exists a tuple of assignments, denoted by $\mathcal{A}^j_t$, such that $(\tilde{S}_j,\mathcal{A}^j_t)_{t=1}^T$ is a feasible solution for $\mathcal{Q}$. We partition set $\tilde{S}_j$ as follows. Let $\tilde{S}_{j,r} = \left\{(i,t_1,t_2)\in \tilde{S}_j~|~ [t_1,t_2]\subseteq [u^j_r,u^j_{r+1}-1]\right\}$. It holds that $\tilde{S}_j = \cup_{r=0}^{k_j-1}\tilde{S}_{j,r}$ as each element $(i,t_1,t_2)$ is contained in exactly one interval $[u^j_r,u^j_{r+1}-1]$. Thus  $\left((\tilde{S}_{j,r},\mathcal{A}^j_t)_{t=u^j_{r}}^{u^j_{r+1}-1}\right)_{r=0}^{k_j-1}$ is a tuple of feasible solutions for the cut instances in $\mathcal{Q}_{U_j}$ such that $(\tilde{S}_{j,r},\mathcal{A}^j_t)_{t=u^j_r}^{u^j_{r+1}-1}$ is a solution for the $r$-th instance. The value of all elements in the defined solutions for the cut instances is
\[
\sum_{r=0}^{k_j-1}\sum_{e\in\tilde{S}_{j,r}}v(e) = \sum_{e\in \tilde{S}_{j}}v(e) = \sum_{e\in X : U_j(e)=\{e\}}v(e) + \sum_{e\in Y} \sum_{e'=(i,t_1,t_2)\in U_j(e) : t_2-t_1\geq \phi}v(e')
\]
Thus, after including the value of gains $g^-$ earned by solutions $\tilde{S}_{j,r}$, we can bound the total value of optimal solutions for all cut instances $\mathcal{Q}_{U_j}$ by 
\begin{equation}\label{eq:lower-bound-cut-instances}
\sum_{e\in X : U_j(e)=\{e\}}v(e) + \sum_{e\in Y} \sum_{\begin{array}{cc}
    {\scriptstyle   e'=(i,t_1,t_2)\in U_j(e) :} \\
    {\scriptstyle  t_2-t_1\geq \phi}
\end{array}}v(e') +\sum_{\begin{array}{cc}
{\scriptstyle t\in [u^j_r+1,u^j_{r+1}-1]: }\\
{\scriptstyle u^j_r\in U_j }
\end{array}}\sum_{
i\notin \tilde{S}_{j,r}(t-1)\cup\tilde{S}_{j,r}(t)}g^-_{i,t}
\end{equation}

We define $\mathcal{B}$ as the total sum of values of optimal solutions for the cut instances $\left(\mathcal{Q}_{U_j}\right)_{j=1}^{\frac{1}{\mu}}$. By utilizing Equation \eqref{eq:lower-bound-cut-instances} we can bound $\mathcal{B}$ as follows.
\begin{equation}\label{eq:bound-three-terms}
\begin{aligned}
    \mathcal{B}&\geq \sum_{j=1}^{\frac{1}{\mu}}\sum_{r=0}^{k_j-1}\left(\sum_{e\in \tilde{S}_{j,r}}v(e) + \sum_{t\in [u^j_r+1,u^j_{r+1}-1]: u^j_r\in U_j}\sum_{i\notin \tilde{S}_{j,r}(t-1)\cup\tilde{S}_{j,r}(t)}g^-_{i,t}\right) \\
    &= \sum_{j=1}^{\frac{1}{\mu}}\sum_{e\in\tilde{S}_j}v(e) + \sum_{j=1}^{\frac{1}{\mu}}\sum_{t\in [2,T]\setminus U_j}\sum_{i\notin \tilde{S}_{j}(t-1)\cup\tilde{S}_{j}(t)}g^-_{i,t}  \\
    &=\sum_{e\in X}\sum_{j\in \left[\frac{1}{\mu}\right] : U_j(e)=\{e\}}v(e) +
    \sum_{e\in Y}\sum_{j=1}^{\frac{1}{\mu}} \sum_{e'=(i,t_1,t_2)\in U_j(e) : t_2-t_1\geq \phi}v(e')\\ 
    & ~~~~~ +\sum_{j=1}^{\frac{1}{\mu}}\sum_{t\in [2,T]\setminus U_j}\sum_{i\notin\tilde{S}_{j}(t-1)\cup\tilde{S}_{j}(t)}g^-_{i,t}
\end{aligned}
\end{equation}
We bound the value of each of the three terms separately by comparing it to the value of the optimal solution.

Consider the first term, value earned from short elements, i.e., elements $e=(i,t_1,t_2)\in X$. It holds that $e\in\tilde{S}_j$ if and only if $U_j(e)=\{e\}$ which means that $(t_1,t_2]\cap U_j=\emptyset$. Since for every $j_1\neq j_2$ it holds that $U_{j_1}\cap U_{j_2} = \{1,T+1\}$ and since there are $\frac{1}{\mu}$ sets of cut point, for each element $e\in X$ it holds that $e\in \tilde{S}_j$ for at least $\frac{1}{\mu} - \frac{\phi}{\epsilon}$ values of $j\in[\frac{1}{\mu}]$. Thus,
\begin{equation}\label{eq:bound-short-elements}
\sum_{e\in X}\sum_{j\in \left[\frac{1}{\mu}\right] : U_j(e)=\{e\}}v(e) \geq \left(\frac{1}{\mu}-\frac{\phi}{\epsilon}\right)\sum_{e\in X}v(e) = \frac{1}{\mu}\left(1-\epsilon\right)\sum_{e\in X}v(e)
\end{equation}

Next, we bound the second term, the value earned from long elements, $e\in Y$. Consider the set of cut points $U_j$. Two operators are applied to each long element. First, it is cut and the subset $U_j(e)$ is defined. Second, short elements are discarded from $U_j(e)$. The resulting subset is $\{(i,t_1,t_2)\in U_j(e) ~|~ t_2-t_1\geq \phi\}$ and therefore we would like to bound the difference
\[ \sum_{e\in Y}\left(v(e)-\sum_{e'\in\{(i,t_1,t_2)\in U_j(e) ~|~ t_2-t_1\geq \phi\}}v(e')\right) \]
Consider an element $e=(i,t_1,t_2)\in Y$ cut by cut points set $U_j$. As shown in Equation \eqref{eq:loss-by-set-of-cuts}, the loss caused by cutting $e$ at cut point $u\in(t_1,t_2]$ is independent of other cuts that are applied to $e$ and is equal to $\ell(e,u)$. Thus we can consider each cut point separately.

As mentioned above, if $e=(i,t_1,t_2)\in Y$, the second operator discards elements $e'\in U_j(e)$ that are short. Since the distance between each pair of cut points in $U_j$ is at least $\frac{1}{\mu}=\frac{\phi}{\epsilon^2}>\phi$, each such short element $e'$ is either $(i,t_1,u)$ or $(i,u,t_2)$ for some unique cut point $u\in U_j$. In addition, it must hold that either $u-t_1<\phi$ or $t_2-u<\phi$. We associate the value lost by discarding $e'$ to this unique cut point $u$.

Let $e=(i,t_1,t_2)\in Y$ and $u\in U_j$ be a cut point such that $u\in(t_1,t_2]$, i.e., $u$ cuts $e$. There are three cases to consider.
\begin{enumerate}
    \item If $u-t_1<\phi$, element $e'=(i,t_1,u-1)\in U_j(e)$ is discarded and a loss of $v(e')$ is associated with $u$ in addition to $\ell(e,u)$. Thus the total loss is at most 
    \begin{align*}
        v(e') + \ell(e,u) &= \sum_{t=t_1}^{u-1}p_t(i) + \sum_{t=t_1+1}^{u-1}g^+_{i,t} -c^+_{i,t_1}-c^-_{i,u-1} + g^+_{i,u}+c^+_{i,u}+c^-_{i,u-1} \leq \\
        &\leq \sum_{t=t_1}^{u-1}p_t(i)+\sum_{t=t_1+1}^{u}g^+_{i,t}+c^+_{i,u-1} \leq \sum_{t=t_1}^{u+\phi-1}p_t(i)+\sum_{t=t_1+1}^{u+\phi-1}g^+_{i,t}
    \end{align*}
    where the equality is due to Equation \eqref{eq:loss-at-cutpoint} and the last inequality is due to the profit-cost ratio.
    
    \item If $t_2-u<\phi$, element $e'=(i,u,t_2)\in U_j(e)$ is discarded and a loss of $v(e')$ is associated with $u$ in addition to $\ell(e,u)$. Thus the total loss is at most  
    \begin{align*}
        v(e') + \ell(e,u) &= \sum_{t=u}^{t_2}p_t(i) + \sum_{t=u+1}^{t_2}g^+_{i,t} -c^+_{i,u}-c^-_{i,t_2} + g^+_{i,u}+c^+_{i,u}+c^-_{i,u-1} \leq \\
        &\leq \sum_{t=u}^{t_2}p_t(i)+\sum_{t=u}^{t_2}g^+_{i,t}+c^-_{i,u-1} \leq \sum_{t=u-\phi}^{t_2}p_t(i)+\sum_{t=u-\phi+1}^{t_2}g^+_{i,t}
    \end{align*}
    where the equality is due to Equation \eqref{eq:loss-at-cutpoint} 
    and the last inequality is due to the profit-cost ratio.
    
    \item If $t_2-u\geq\phi$ and $u-t_1\geq\phi$, no elements are discarded from $U_j(e)$. Thus the only loss is $\ell(e,u)$ and can be bounded by
    \[ \ell(e,u) = g^+_{i,u}+c^+_{i,u}+c^-_{i,u-1} \leq \sum_{t=u-\phi}^{u+\phi-1}p_t(i)+\sum_{t=u-\phi+1}^{u+\phi-1}g^+_{i,t} \]
\end{enumerate}
Overall we can bound the loss induced by cutting long elements at cut points $U_j$ (due to loss $\ell(e,u)$ and discarded short elements) by
\[ \sum_{(i,t_1,t_2)\in Y}\sum_{u\in (t_1,t_2]\cap U_j}\left(\sum_{t=\max\{t_1,u-\phi\}}^{\min\{t_2,u+\phi-1\}}p_t(i)+\sum_{t=\max\{t_1+1,u-\phi+1\}}^{\min\{t_2,u+\phi-1\}}g^+_{i,t}\right) \]
This means that the total value gained from elements that were originally in $Y$ is
\begin{align*}
    &\sum_{e\in Y}\sum_{j=1}^{\frac{1}{\mu}} \sum_{e'=(i,t_1,t_2)\in u(e,U_j) : t_2-t_1\geq \phi}v(e') \geq \\
    &\frac{1}{\mu}\sum_{e\in Y}v(e) - \sum_{j=1}^{\frac{1}{\mu}}\sum_{(i,t_1,t_2)\in Y}\sum_{u\in [t_1+1,t_2]\cap U_j}\left(\sum_{t=\max\{t_1,u-\phi\}}^{\min\{t_2,u+\phi-1\}}p_t(i)+\sum_{t=\max\{t_1+1,u-\phi+1\}}^{\min\{t_2,u+\phi-1\}}g^+_{i,t}\right) \geq \\
    & \frac{1}{\mu}\sum_{e\in Y}v(e) - \sum_{(i,t_1,t_2)\in Y}\sum_{u\in[t_1+1,t_2]}\left(\sum_{t=\max\{t_1,u-\phi\}}^{\min\{t_2,u+\phi-1\}}p_t(i)+\sum_{t=\max\{t_1+1,u-\phi+1\}}^{\min\{t_2,u+\phi-1\}}g^+_{i,t}\right) \geq \\
    & \frac{1}{\mu}\sum_{e\in Y}v(e) -2\phi\sum_{(i,t_1,t_2)\in Y}\left(\sum_{t\in[t_1,t_2]}p_t(i)+\sum_{t\in[t_1+1,t_2]}g^+_{i,t}\right)
\end{align*}
where the second inequality follows from the fact that for every $j_1\neq j_2$ it holds that $U_{j_1}\cap U_{j_2}=\{1,T+1\}$. The last inequality is due to the fact that the profit and gains of element $(i,t_1,t_2)$ is lost at stage $t\in[t_1,t_2]$ if it is cut by a cut point $u$ such that $|u-t|\leq\phi$. Thus, its value is lost in at most $2\phi$ instances. Due to the profit-cost ratio, for each long element $e=(i,t_1,t_2)\in Y$ it holds that
\[ c^+_{i,t_1}+c^-_{i,t_2}\leq 2\phi\cdot\frac{\sum_{t=t_1}^{t_2}p_t(i)+\sum_{t=t_1+1}^{t_2}g^+_{i,t}}{t_2-t_1} \leq \epsilon\cdot\sum_{t=t_1}^{t_2}p_t(i)+\epsilon\cdot\sum_{t=t_1+1}^{t_2}g^+_{i,t} \]
By substituting $4\phi(c^+_{i,t_1}+c^-_{i,t_2})\leq 4\phi\epsilon\left(\sum_{t=t_1}^{t_2}p_t(i)+\sum_{t=t_1+1}^{t_2}g^+_{i,t}\right)$ we get that
\begin{align*}
    & \frac{1}{\mu}\sum_{e\in Y}v(e) -2\phi\sum_{(i,t_1,t_2)\in Y}\left(\sum_{t=t_1}^{t_2}p_t(i)+\sum_{t=t_1+1}^{t_2}g^+_{i,t}\right) = \\
    & \sum_{(i,t_1,t_2)\in Y}\left(\frac{1}{\mu}\left(\sum_{t=t_1}^{t_2}p_t(i) + \sum_{t=t_1+1}^{t_2}g^+_{i,t} - c^+_{i,t_1}-c^-_{i,t_2}\right) -2\phi\left(\sum_{t=t_1}^{t_2}p_t(i)+\sum_{t=t_1+1}^{t_2}g^+_{i,t}\right)\right) \geq \\
    & \sum_{(i,t_1,t_2)\in Y}\left(\left(\frac{1}{\mu}-2\phi-4\phi\epsilon\right)\left(\sum_{t=t_1}^{t_2}p_t(i) + \sum_{t=t_1+1}^{t_2}g^+_{i,t}\right) - \left(\frac{1}{\mu}-4\phi\right)\left(c^+_{i,t_1}+c^-_{i,t_2}\right)\right) \geq \\
    & \left(\frac{1}{\mu}-2\phi-4\phi\epsilon\right)\sum_{e\in Y}v(e) \geq \left(\frac{1}{\mu}-\frac{\phi}{\epsilon}\right)\sum_{e\in Y}v(e)
\end{align*}
where the last inequality follows the fact that $\epsilon<\frac{1}{4}$. Overall, we get that
\begin{equation}\label{eq:bound-long-elements}
    \sum_{e\in Y}\sum_{j=1}^{\frac{1}{\mu}} \sum_{e'=(i,t_1,t_2)\in U_j(e) : t_2-t_1\geq \phi}v(e')  \geq \left(\frac{1}{\mu}-\frac{\phi}{\epsilon}\right)\sum_{e\in Y}v(e)
\end{equation}

Lastly, we bound the third term, gains $g^-$ earned in all cut instance solutions. Consider a cut point set $U_j$ and gain $g^-_{i,t}$ earned in solution $\tilde{S}^*$, i.e., $i\notin S^*_{t-1}\cup S^*_t$. Therefore, any element $(i,t_1,t_2)$ such that $t\in [t_1,t_2]$ or $t-1\in[t_1,t_2]$ is not in $\tilde{S}_j$. Thus, unless $t\in U_j$, gain $g^-_{i,t}$ is earned in $\tilde{S}_j$ and in some solution  $\tilde{S}_{j,r}\subseteq \tilde{S}_j$ such that $t\in [u^j_r,u^j_{r+1}-1]$ and $u^j_r\in U_j$. Again, we can use the fact that for every $j_1\neq j_2$ it holds that $U_{j_1}\cap U_{j_2}=\{1,T+1\}$ and get that
\begin{equation}\label{eq:bound-gains-term}
\begin{aligned}
    \sum_{j=1}^{\frac{1}{\mu}}\sum_{t\in [2,T]\setminus U_j}\sum_{i\notin \tilde{S}_j(t-1)\cap\tilde{S}_j(t)}g^-_{i,t} \geq 
    &\left(\frac{1}{\mu}-1\right)\sum_{t=1}^T\sum_{i\notin \tilde{S}^*(t-1)\cap\tilde{S}^*(t)}g^-_{i,t}
\end{aligned}
\end{equation}

By substituting inequalities \eqref{eq:bound-short-elements},\eqref{eq:bound-long-elements} and \eqref{eq:bound-gains-term} in Inequality \eqref{eq:bound-three-terms} we get that
\[ \mathcal{B} \geq \left(\frac{1}{\mu}-\frac{\phi}{\epsilon}\right)\left(\sum_{e\in \tilde{S}^*}v(e)+\sum_{t\in [2,T]}\sum_{i\notin \tilde{S}^*(t-1)\cap\tilde{S}^*(t)}g^-_{i,t}\right) = \left(\frac{1}{\mu}-\frac{\phi}{\epsilon}\right)f_\mathcal{Q}((S^*_t)_{t=1}^T) \]

For each set of values, their average is smaller or equal to their maximum value. Thus there must exist at least one set of cut points $U_{j^*}$ such that the sum of values of the solutions $\left(\tilde{S}_{j^*,r}\right)_{r=0}^{k_j-1}$ for its cut instances, $\mathcal{Q}_{U_j^*}$, is at least $\mu\cdot\mathcal{B}$, the average value of a set of solutions for a set of cut instance $\mathcal{Q}_{U_j}$ (for $j=1,\ldots,\frac{1}{\mu}$). We get that
\[\sum_{e\in\tilde{S}_{j^*}}v(e) + \sum_{t\in [2,T]}\sum_{i\notin \tilde{S}_{j^*}(t-1)\cap\tilde{S}_{j^*}(t)}g^-_{i,t} \geq \mu\mathcal{B} =  \left(1-\epsilon\right)f_\mathcal{Q}\left((S^*_t)_{t=1}^T\right)\]
At iteration $j^*$, in which Algorithm \ref{alg:break-time-horizon} considers the cut instances $\mathcal{Q}_{U_{j^*}}$, algorithm $A$ provides an approximate solution for each cut instance. Thus the value of the solution returned by $A$ for the $r$-th cut instance is at least 
\[ \alpha\cdot\left(\sum_{e\in \tilde{S}_{j^*,r}}v(e) + \sum_{t\in[u^j_r+1,u^j_{r+1}-1]:u^j_r\in U_j}\sum_{i\notin\tilde{S}_{j^*,r}(t-1)\cup\tilde{S}_{j^*,r}(t)}g^-_{i,t} \right) \]
Summing over all cut instances in $\mathcal{Q}_{U_{j^*}}$ provides a solution with value at least
\[ \alpha\cdot\left(\sum_{e\in\tilde{S}_{j^*}}v(e) + \sum_{t\in [2,T]}\sum_{i\notin \tilde{S}_{j^*}(t-1)\cap\tilde{S}_{j^*}(t)}g^-_{i,t}\right) \geq\left(1-\epsilon\right)\cdot\alpha\cdot f_\mathcal{Q}\left((S^*_t)_{t=1}^T\right) \]
\end{proof}
The correctness of Theorem \ref{thm:ptas-for-modular} follows immediately from Theorem \ref{lem:cutting-algorithm} and Lemma \ref{lem:ptas-for-fixed-horizon}.

\section{Hardness Results}\label{sec:hardness}
In this section we present two hardness results for $1$-GMK. First, we show no constant approximation ratio exists for $1$-GMK (with unbounded profit-cost ratio), even if there is only one bin per stage. Then, we show that even if we wither down the model by removing the change costs, limiting the time horizon length to $T=2$, and only having a single bin per stage, the problem still does not admit an EPTAS.

The above results are proved by showing an approximation preserving reduction from  $d$-Dimensional Knapsack ($d$-KP) and Multidimensional Knapsack. For $d\in \mathbb{N}$, in $d$-KP we are given a set of items $I$, each equipped with a profit $p_i$, as well as a $d$-dimensional weight vector $\bar{w}^i\in [0,1]^d$. We denote $j$-th coordinate of $\bar{w}^i$ by $\bar{w}^i_j$. In addition, we are given a single bin equipped with a $d$-dimensional capacity vector $\bar{W}\in\mathbb{R}_{\geq  0}^d$. A subset $S\subseteq I$ is a feasible solution if $\sum_{i\in S}\bar{w}^i\leq \bar{W}$. The objective is to find a feasible solution $S$ which maximizes $\sum_{i\in S}p_i$. 

The Multidimensional Knapsack problem is the generalization of $d$-KP in which $d$ is not fixed. That is, the input for the problem is a $d$-KP instance for some $d\in \mathbb{N}$. The  solutions and their values are the solution and values of the $d$-KP instance.

Note that $d$-KP is a special case of $d$-MKCP, where the set of MKCs is $\mathcal{K}=(K_j)_{j=1}^d$. The $j$-th MKC is $K_j=(w_j,B_j,W_j)$, where $w_j(i)=\bar{w}_j^i, B_j=\{b\}$ and $W_j(b)=\bar{W}_j$, where $\bar{W}_j$ is the $j$-th coordinate of the capacity vector $\bar{W}$. Finally, the profit function $p:I\rightarrow\mathbb{R}_{\geq  0}$ is defined as $p(i) = p_i$ for any $i\in I$. For simplicity we will use this notation for $d$-KP and Multidimensional Knapsack throughout this section.

\label{sec:d-dim-kp-hardness}



\begin{lemma}\label{lem:d-kp-reduction}
There is an approximation preserving reduction from the Multidimensional Knapsack problem to $1$-GMK with a single bin in each stage.
\end{lemma}

\begin{proof}
Let $\mathcal{Q}=(I,\mathcal{K},p)$ be an instance of  Multidimensional Knapsack, where $\mathcal{K} = (K_j)_{j=1}^d$. 
We define an instance of $1$-GMK as follows. Define $T=d$, and for $j=1,\ldots,d$ define  $\mathcal{P}_j = (I,(K_j),h)$ with $h(i)=\frac{p(i)}{d}$ for all $i\in I$. The gains vectors are defined as zero vectors, $g^+=g^-=\vv{0}$. Finally, we define the change cost vectors. For all $i\in I$ we set
\[c^+_{i,t}=
\begin{cases}
    p(i) & t\in[2,d] \\
    0 & \text{otherwise.}
\end{cases},~~~
c^-_{i,t}=
\begin{cases}
    p(i) & t\in[1,d-1] \\
    0 & \text{otherwise.}
\end{cases}
\]
The tuple $\tilde{\mathcal{Q}}=\left((\mathcal{P}_t)_{t=1}^d, g^+,g^-,c^+,c^-\right)$ is a $1$-GMK instance with time horizon $T=d$. 


Let $(S,\mathcal{A})$ be a feasible solution for $\mathcal{Q}$, where $\mathcal{A} = (A_j)_{j=1}^d$. We can easily construct a solution for  $\tilde{\mathcal{Q}}$ by setting $\mathcal{A}_j=(A_j)$ for $j=1,\ldots,d$. Then, $(S_t,\mathcal{A}_t)_{t=1}^d$, where $S_j=S$ for $j=1,\ldots,d$, is a solution for $\tilde{\mathcal{Q}}$. Note that all items are either assigned or not assigned in all stages. Thus the value of the solution is
\[ f_{\tilde{\mathcal{Q}}}((S_t)_{t=1}^d) = \sum_{t=1}^d h(S_t)-\sum_{t=1}^T\left(\sum_{i\in S_t\setminus S_{t+1}}c^-_{i,t}+\sum_{i\in S_{t}\setminus S_{t-1}}c^+_{i,t}\right) = d\cdot h(S) = p(S) \]
The solution is also feasible as for every $j\in[d]$ it holds that $A_j$ is a feasible assignments for MKC $K_j$.

Next, let $(S_t,\mathcal{A}_t)_{t=1}^d$ be a feasible solution for $\tilde{\mathcal{Q}}$, where $\mathcal{A}_j = (\tilde{A}_j)$ for $j=1,\ldots,d$. 
For $j=1,\ldots,d$ let $B_j=\{b_j\}$. We define the selected items set as $S=\bigcap_{j\in[d]}S_j$ and define the assignments accordingly, $A_j(b_j)=S$ for $j=1,\ldots,d$ and $\mathcal{A}=(A_j)_{j=1}^d$. Consider some $j\in[s]$, the assignment $A_j$ is feasible as $A_j(b_j)\subseteq \tilde{A}_j(b_j)$ and
\[ \sum_{i\in A_j(b_j)}w_j(i)\leq \sum_{i\in \tilde{A}_j(b_j)}w_j(i)\leq W_j(b_j) \]
The value of the solution is 
\[ p(S) = d\cdot h(S) \geq \sum_{t=1}^d h(S_t) - \sum_{t=1}^{T}\left(\sum_{i\in S_t\setminus S_{t+1}}c^-_{i,t} + \sum_{i\in S_{t}\setminus S_{t-1}}c^+_{i,t}\right) = f_{\tilde{\mathcal{Q}}}((S_t)_{t=1}^d) \]
where the inequality follows from the construction of $\tilde{Q}$. Furthermore, note that $S$ can be constructed in polynomial time, which concludes the proof.
\end{proof}

In 
\cite{chekuri2004multidimensional} Chekuri and Kahanna showed that Multidimensional Knapsack does not admit any constant approximation ratio unless $NP=ZPP$. 
Theorem \ref{thm:phi-hardness} follows immediately from the hardness result of~\cite{chekuri2004multidimensional} and Lemma~\ref{lem:d-kp-reduction}.
\label{sec:2-dim-kp-hardness}


We  now proceed to the second hardness result. 
\begin{lemma}\label{lem:2kp-reduction}
There is an approximation preserving reduction from $2$-Dimensional Knapsack problem to $1$-GMK with time horizon $T=2$, no change costs and a single bin per stage.
\end{lemma}

\begin{proof}
Let $\mathcal{Q}=(I,\mathcal{K},p)$ be an instance of $2$-dimensional knapsack, where $\mathcal{K} = (K_1,K_2)$. Also, since $p$ is modular, it holds that $p(S) = \sum_{i\in S}p_i$. 

We define an instance of $1$-GMK as follows. Set $T=2$, $\mathcal{P}_1 = (I,(K_1),h)$ and $\mathcal{P}_2 = (I,(K_2),h)$, where $h$ is the zero function, i.e., $h:I\rightarrow\{0\}$ such that $\forall i\in I$ it holds that $h(i)=0$. Since there are only two stages, gains exists only for stage for $t=2$. Set $g^+_{i,2}=p(i)$ and $g^-_{i,2}=0$ for each item $i\in I$. Finally, we set the change cost vectors as $c^+=c^-=\vv{0}$. The tuple $\tilde{\mathcal{Q}}=((\mathcal{P}_1,\mathcal{P}_2), g^+,g^-,c^+,c^-)$ is a $1$-GMK instance with time horizon $T=2$. Note that since all profits, change costs and gains $g^-$ are zero we can write the objective function as
\[f_{\tilde{\mathcal{Q}}}\left((S_1,S_2)\right) = \sum_{i\in S_1\cap S_2}g^+_{i,2}.\]

Let $(S,\mathcal{A})$ be a feasible solution for $\mathcal{Q}$, where $\mathcal{A} = (A_1,A_2)$. We can easily construct a solution for the $\tilde{\mathcal{Q}}$ by setting $\mathcal{A}_1=(A_1)$ and $\mathcal{A}_2=(A_2)$. Then, $(S_t,\mathcal{A}_t)_{t=1}^2$, where $S_1=S_2=S$, is a solution for $\tilde{\mathcal{Q}}$. Note that all items are either assigned in both stages or not assigned in both stages. Thus the value of the solution is
\[ f_{\tilde{\mathcal{Q}}}((S,S)) = \sum_{i\in S_1\cap S_2} g^+_{i,2} = \sum_{i\in S} g^+_{i,2} = \sum_{i\in S}p_i = p(S) \]
The solution is also feasible as $A_1$ and $A_2$ are feasible assignments of $K_1$ and $K_2$ (respectively).

Next, let $(S_t,\mathcal{A}_t)_{t=1}^2$ be a feasible solution for $\tilde{\mathcal{Q}}$, where $\mathcal{A}_1 = (\tilde{A}_1)$ and $\mathcal{A}_2=(\tilde{A}_2)$. 
Let $B_1=\{b_1\}$ and $B_2=\{b_2\}$. We define the selected items set as $S=S_1\cap S_2$ and define the assignments accordingly, $A_1(b_1)=A_2(b_2) = S$ and $\mathcal{A}=(A_1,A_2)$. Assignment $A_1$ is feasible as $A_1(b_1)\subseteq \tilde{A}_1(b_1)$ and
\[ \sum_{i\in A_1(b_1)}w_1(i)\leq \sum_{i\in \tilde{A}_1(b_1)}w_1(i)\leq W_1(b_1) \]
A similar statement shows that assignment $A_2$ is feasible as well. The value of the solution is 
\[ p(S) = \sum_{i\in S}p(i) =  \sum_{i\in S}g^+_{i,2} =\sum_{i\in S_1\cap S_2} g^+_{i,2} = f_{\tilde{\mathcal{Q}}}((S,S)) \]
which concludes the proof.
\end{proof}

In \cite{kulik2010there} Kulik and Shachnai showed that there is no EPTAS for $2$-KP unless $W[1]=FPT$. Theorem~\ref{thm:no-eptas} follows immediately from the hardness result of \cite{kulik2010there} and Lemma~\ref{lem:2kp-reduction}.


\bibliographystyle{plainurl}
\bibliography{bib}

\appendix

\newpage
\section{Omitted Proofs and Definition}\label{app:omitted-proofs}

\begin{definition}\label{def:approx-reduction}
Let $\Pi_1,\Pi_2$ be two maximization problems. An {\em approximation factor preserving reduction} from $\Pi_1$ to $\Pi_2$ consists of two polynomial time algorithms $f,g$ such that for any two instances $I_1$ of problem $\Pi_1$ and $I_2=f(I_1)$ of problem $\Pi_2$ it holds that
\begin{itemize}
    \item $I_2\in \Pi_2$ and $OPT_{\Pi_2}(I_2)\geq OPT_{\Pi_1}(I_1)$.
    \item for any solution $s_2$ for $I_2$, solution $s_1=g(I_1,s_2)$ is a solution for $I_1$ and $obj_{\Pi_1}(I_1,s_1)\geq obj_{\Pi_2}(I_2,s_2)$. 
\end{itemize}
where $OPT_{\Pi}(I)$ is the value of an optimal solution for instance $I$ of problem $\Pi$, and $obj_{\Pi}(I,s)$ is the value of solution $s$ for instance $I$ of problem $\Pi$.
\end{definition}

\begin{proof}[Proof of Lemma \ref{lem:modular-reduction-dir2}]
Let $\mathcal{Q}=\left((\mathcal{P}_t)_{t=1}^T,g^+,g^-,c^+, c^-\right)$ be an instance of $d$-GMK, where $\mathcal{P}_{t} = \left(I,\mathcal{K}_{t},p_t\right)$ and $\mathcal{K}_t = (K_{t,j})_{j=1}^{d_t}$. Also, let $R(\mathcal{Q})=\left(E,\tilde{\mathcal{K}},\tilde{p},\mathcal{I}\right)$ be the reduced instance of $\mathcal{Q}$. and $\left(S, \left(\tilde{A}_{t,j}\right)_{t\in [T],~j\in [d]}\right)$ be a feasible solution for $R(\mathcal{Q})$. We define the solution $(S_t,\mathcal{A}_t)_{t=1}^T$ for $\mathcal{Q}$ as follows. For every stage $t$ we set $S_t = \{i\in I~|~ \exists(i,D)\in S : t\in D\}$. For every $t=1,\ldots,T$, $j=1,\ldots,d_t$ and bin $b\in B_{t,j}$ (the set of bins in MKC $K_{t,j}$) let $A_{t,j}(b) = \{i\in I~|~ \exists (i,D)\in \tilde{A}_{t,j}(b) : t\in D\}$. Observe that the sets $(S_t)_{t=1}^{T}$ and assignments $(A_{t,j})_{t\in[T],~j\in[d_t]}$ can be constructed in polynomial time as at most $|I|$ elements can be chosen due to the matroid constraint. The assignment $A_{t,j}$ is an assignment of $S_t$ since
\[ S_t = \{i\in I~|~\exists(i,D)\in S:t\in D\} = \bigcup_{b\in B_{t,j}}\{i\in I~|~(i,D)\in \tilde{A}_{t,j}(b):t\in D\} = \bigcup_{b\in B_{t,j}}A_{t,j}(b),\]
where the second equality follows the feasibility of the solution for $R(\mathcal{Q})$. In addition, $A_{t,j}$ is a feasible assignment for MKC $K_{t,j}$ since for every bin $b\in B_{t,j}$ it holds that
\[ \sum_{i\in A_{t,j}(b)}w_{t,j}(i) = \sum_{(i,D)\in\tilde{A}_{t,j}(b) :t\in D}\tilde{w}_{t,j}((i,D)) = \sum_{(i,D)\in\tilde{A}_{t,j}(b)}\tilde{w}_{t,j}((i,D)) \leq W_{t,j}(b) \]
Thus $\left(S_t,\mathcal{A}_t\right)_{t=1}^T$ is a feasible solution for $\mathcal{Q}$.

Lastly, consider the value of the solution for $\mathcal{Q}$. It holds that
\[
\begin{aligned}
    &f_{\mathcal{Q}}\left((S_t)_{t=1}^T\right)= \\
    &\sum_{t=1}^T\sum_{i\in S_t} p_t(i) + \sum_{t=2}^{T}\left( \sum_{i\in S_{t-1}\cap S_{t}} g^+_{i,t} + \sum_{i\notin S_{t-1}\cup S_{t}} g^-_{i,t}\right) - \sum_{t=1}^{T}\left( \sum_{i\in S_{t}\setminus S_{t-1}} c^+_{i,t} + \sum_{i\in S_t\setminus S_{t+1}} c^-_{i,t} \right) = \\
    &\sum_{(i,D)\in S}\left(\sum_{t\in D}p_t(i)+\sum_{t\in D : t-1\in D}g^+_{i,t} +\sum_{t\notin D : t-1\notin D}g^-_{i,t} - \sum_{t\in D : t-1\notin D}c^+_{i,t} -\sum_{t\in D : t+1\notin D}c^-_{i,t} \right) = \\ 
    &\tilde{p}(S)
\end{aligned}
\]
\end{proof}

\begin{proof}[Proof of Corollary \ref{cor:cut-solution}]

Let $\mathcal{Q}$ be an instance of $d$-GMK, $U$ be a set of cut points. Also, let $\mathcal{Q}_U = (q_j)_{j=0}^{k-1} = \left((\mathcal{P}_t)_{t=u_j}^{u_{j+1}-1},g^+,g^-,c^+,c^-\right)_{j=0}^{k-1}$ be the corresponding tuple of cut instances, and $\left((S_t,\mathcal{A}_t)_{t=u_j}^{u_{j+1}-1}\right)_{j=0}^{k-1}$ be a tuple of feasible solutions for the cut instances. 

We define the solution $\left(S_t,\mathcal{A}_t\right)_{t=1}^T$ for $\mathcal{Q}$. It is easy to see that the assignments $\mathcal{A}_t$ to $\mathcal{K}_t$ are all feasible assignments of $S_t$. In addition, it holds that
\[
\begin{aligned}
    &f_{\mathcal{Q}}\left((S_t)_{t=1}^T\right) = \\
    &\sum_{t=1}^Tp_t(S_t) + \sum_{t=2}^{T}\left(\sum_{i\in S_{t-1}\cap S_{t}}g^+_{i,t}+\sum_{i\notin S_{t-1}\cup S_{t}}g^-_{i,t}\right) - \sum_{t=1}^{T}\left(\sum_{i\in S_t\setminus S_{t+1}}c^-_{i,t}+\sum_{i\in S_{t}\setminus S_{t-1}}c^+_{i,t}\right) \geq \\
    & \sum_{j=0}^{k-1}\left(\sum_{t=u_j}^{u_{j+1}-1}p_t(S_t) + \sum_{t=u_j+1}^{u_{j+1}-1}\left(\sum_{i\in S_{t-1}\cap S_{t}}g^+_{i,t}+\sum_{i\notin S_{t-1}\cup S_{t}}g^-_{i,t}\right) \right. \\
    & ~~~~~~~~ -\left.\sum_{t=u_j+1}^{u_{j+1}-1}\left(\sum_{i\in S_t\setminus S_{t+1}}c^-_{i,t}+\sum_{i\in S_{t}\setminus S_{t-1}}c^+_{i,t}\right) - 
    \sum_{i\in S_{u_{j+1}-1}}c^-_{i,u_{j+1}} - \sum_{i\in S_{u_j}}c^+_{i,u_j} \right) = \\
    &\sum_{j=0}^{k-1}f_{q_j}\left( (S_t)_{t=u_j}^{u_{j+1}-1} \right)
\end{aligned}
\]
where $f_{q_j}$ is the objective functions of cut instance $q_j$. This proves that a cut solution has a higher value than the sum of solutions for cut instance from which it was created.
\end{proof}
\section{Approximation Scheme for Submodular \texorpdfstring{$d$}--GMK}
\label{app:submodular}

In this section we define the Submodular $d$-MKCP problem and provide approximation algorithm for it. We start by defining several properties of set functions.

A set function $p:2^I\rightarrow \mathbb{R}_{\geq 0}$ is said to be submodular if for every $A\subseteq B\subseteq I$ and $i\in I\setminus B$ it holds that $p(A\cup\{i\})-f(A)\geq p(B\cup\{i\})-f(B)$. This property is sometimes referred to as diminishing returns. Set function $p$ is said to be monotone if for all $A\subseteq B\subseteq I$ it holds that $p(A)\leq p(B)$. Lastly, $p$ is said to be non-negative if for every $A\subseteq I$ it holds that $f(A)\geq 0$.

In Submodular $d$-MKCP we are given a tuple $\left( I, \mathcal{K},p \right)$, where $I$ is a set of items, $\mathcal{K}$ is a tuple of $d$ MKCs and $p:2^I\rightarrow \mathbb{R}_{\geq 0}$ is a non-negative, monotone and submodular set function, defining the profit of a subset $S\subseteq I$. A feasible solution for Submodular $d$-MKCP is a set $S\subseteq I$ and a tuple of feasible assignments ${\mathcal{A}}$ (w.r.t $\mathcal{K}$) of $S$. The goal is to find a feasible solution that maximizes $p(S)$. 

Submodular $d$-GMK is the multistage model of Submodular $d$-MKCP. The problem is defined over a time horizon of $T$ stages as follows. An instance of the problem is a tuple $\left((\mathcal{P}_t)_{t=1}^T, g^+,g^-\right)$, where  $\mathcal{P}_t = \left(I,\mathcal{K}_t,p_t\right)$ is a Submodular $d_t$-MKCP instance with $d_t\leq d$ for $t\in[T]$ and $g^+,g^-\in \mathbb{R}_+^{I\times[2,T]}$ are the gains vectors. We use $g^+_{i,t}$ and $g^-_{i,t}$ to denote the gain of item $i$ at stage $t$. 

A feasible solution for Submodular $d$-GMK is a tuple $(S_t,\mathcal{A}_t)_{t=1}^T$, where $(S_t,\mathcal{A}_t)$ is a feasible solution for $\mathcal{P}_t$ (note that $\mathcal{A}_t$ is a tuple of assignments of $S_t$). The objective function of instance $\mathcal{Q}$ is denoted by $f_{\mathcal{Q}}:I^T\rightarrow \mathbb{R}$, where
\[ f_{\mathcal{Q}}\left((S_t)_{t=1}^T\right)=\sum_{t=1}^T p_t(S_t) + \sum_{t=2}^{T}\left(\sum_{i\in S_{t-1}\cap S_{t}}g^+_{i,t} + \sum_{i\notin S_{t-1}\cup S_{t}}g^-_{i,t}\right) \]
The goal is to find a feasible solution that maximizes the objective function $f_\mathcal{Q}$.

\subsection{Bounded Time Horizon}\label{app:submodular-reduction}

In this section we provide a reduction from an instance of Submodular $q$-GMK to a generalization of Submodular $d$-MKCP (for specific values of $d$ and $q$). The generalization was presented in \cite{fairstein2020tight} and is called Submodular $d$-MKCP With A Matroid Constraint (Submodular $d$-MKCP+) and is defined by a tuple $(I,\mathcal{K},p,\mathcal{I})$, where $(I,\mathcal{K},p)$ forms an instance of Submodular $d$-MKCP. Also, the set $\mathcal{I}\subseteq 2^I$ defines a matroid constraint. A feasible solution for Submodular $d$-MKCP+ is a set $S\in\mathcal{I}$ and a tuple of feasible assignments ${\mathcal{A}}$ (w.r.t $\mathcal{K}$) of $S$. The goal is to find a feasible solution which maximizes $p(S)$.

The following claim is essential in proving the reduction from $q$-GMK to $d$-MKCP+.

\begin{claim}\label{submodular-extension}
	Let $f:2^I\rightarrow\mathbb{R}_{\geq 0}$ be a non-negative, monotone and submodular function, and let
	$E \subseteq I \times 2^X$
	for some set $X$ (each element of $E$ is a pair $(e,H)$ with $e\in I$ and $H\subseteq X$). Then for any $h\in X$ the function $g:2^{E}\rightarrow\mathbb{R}_{\geq 0}$ defined by $g(A) =f\left(\{i~|~ \exists (i,H)\in A : h\in H\}\right)$ is non-negative, monotone and submodular.
\end{claim}

\begin{proof}

Let $S,T\subseteq E$, then
\begin{equation*}
\begin{aligned}
g(S)+g(T)&=
f\left(\left\{i~\middle|~\exists(i,H)\in S:~h\in H\right\}\right)+
f\left(\left\{i~\middle|~\exists(i,H)\in T:~h\in H\right\}\right)\\
&\geq
f\left(\left\{i~\middle|~\exists(i,H)\in S:~h\in H\right\}\cup \left\{i~\middle|~\exists(i,H)\in T:~h\in H\right\}\right) \\&~~~+ f\left(\left\{i~\middle|~\exists(i,H)\in S:~h\in H\right\}\cap \left\{i~\middle|~\exists(i,H)\in T:~h\in H\right\}\right)\\
&\geq f\left(\left\{i~\middle|~\exists(i,H)\in S\cup T:~h\in H\right\}\right)+
f\left(\left\{i~\middle|~\exists(i,H)\in S\cap T:~h\in H\right\}\right)\\
&=g(S\cup T)+g(S\cap T)
\end{aligned}
\end{equation*}
where the first inequality is by submodularity, and the second inequality uses 
\[\{i~|~\exists(i,D)\in S:~t\in D\}\cap \{i~|~\exists(i,D)\in T:~t\in D\}\subseteq \{i~|~\exists(i,D)\in S\cup T:~t\in D\}\] 
and the monontonicity of $f$. Thus $g$
 is submodular.
 
Let $S\subseteq T\subseteq E$, then
\[g(S)= f\left(\left\{i~\middle|~\exists(i,H)\in S:~h\in H\right\}\right) \leq f\left(\left\{i~\middle|~\exists(i,H)\in T:~h\in H\right\}\right)=g(T),\]
where the in equality holds since $f$ is monotone. Thus $g$ is monotone as well.

Finally, for any $S\subseteq E$, 
\[g(S)= f\left(\left\{i~\middle|~\exists(i,H)\in S:~h\in H\right\}\right)\geq 0,\]
as $f$ is non-negative. Thus $g$ is non-negative. 
\end{proof}

We are now ready to provide the reduction. The following definition presents the construction of the reduction.

\begin{definition}\label{def:submodular-reduction}
Let $\mathcal{Q}=\left((\mathcal{P}_t)_{t=1}^T,g^+,g^-\right)$ be an instance of Submodular $d$-GMK, where $\mathcal{P}_{t} = \left(I,\mathcal{K}_{t},p_t\right)$ and $\mathcal{K}_t = (K_{t,j})_{j=1}^{d_t}$. Given instance $\mathcal{Q}$, the operator $R$ returns the Submodular $dT$-MKCP+ instance $R(\mathcal{Q})=\left(E,\tilde{\mathcal{K}},\tilde{p},\mathcal{I}\right)$ where
\begin{itemize}
    \item $E=I\times 2^{[T]}$
    \item $\mathcal{I}=\left\{ S\subseteq E~\middle|~ \forall i\in I :~\left|S\cap \left(\{i\}\times2^{[T]}\right)\right|=1 \right\}$
    \item For $t=1,...,T, j\leq d_t$ set MKC $\tilde{K}_{t,j}=(\tilde{w}_{t,j},B_{t,j},W_{t,j})$ over $E$, where \\ $K_{t,j} = (w_{t,j},B_{t,j},W_{t,j})$ and
    \[ \tilde{w}_{t,j}((i,D)) = 
\begin{cases}
	w_{t,j}(i) & t\in D\\
	0 & \textnormal{otherwise}
\end{cases} \]
    \item For $t=1,...,T, d_t<j\leq d$ set MKC $\tilde{K}_{t,j}= (w_0,\{b\},W_0)$ over $E$, where $w_0:2^E\rightarrow \{0\}$, $W_0(b)=0$ and $b$ is an arbitrary bin (object).
    \item $\tilde{\mathcal{K}} = \left( \tilde{K}_{t,j} \right)_{t\in[T],j\in [d]}$.
    \item For $t=1,...,T$: set
    $ \tilde{p}_t(S)= 
        p_t\left(\{i\in I~|~ \exists (i,D)\in S : t\in D\}\right)
     $
    \item The objective function $\tilde{f}$ is defined as follows.
    \[ \tilde{p}(S) = \sum_{t=1}^T \tilde{p}_t(S) + \sum_{(i,D)\in S}\left(\sum_{t\in D :~ t-1\in D}g^+_{i,t} +\sum_{t\notin D :~ t-1\notin D}g^-_{i,t} \right) \]
\end{itemize}
\end{definition}

Each element $(i,D)\in E$ states the subset of stages in which item $i$ is assigned. I.e., $i$ is only assigned in stages $t\in D$. Thus any solution should include at most one element $(i,D)$ for each $i\in I$. This constraint is fully captured by the partition matroid constraint defined by the set of independent sets $\mathcal{I}$. Finally, if an element is $(i,D)$ is selected, we must assign $i$ in each MKC $K_{t,j}$ for $j\in [d_t], t\in D$. This is captured by the weight function $\tilde{w}$, as an element $(i,D)$ weighs $w_{t,j}(i)$ if and only if $t\in D$ (otherwise its weight is zero and it can be assigned for ``free'').

\begin{lemma}
The reduced instance $R(\mathcal{Q})$ of a Submodular $d$-GMK instance $\mathcal{Q}$ over time horizon $T$ is a valid Submodular $dT$-MKCP+ instance.
\end{lemma}

\begin{proof}
Let $\mathcal{Q}=\left((\mathcal{P}_t)_{t=1}^T,g^+,g^-\right)$ be an instance of Submodular $d$-GMK, where $\mathcal{P}_{t} = \left(I,\mathcal{K}_{t},p_t\right)$ and $\mathcal{K}_t = (K_{t,j})_{j=1}^{d_t}$. Also, let $R(\mathcal{Q})=\left(E,\tilde{\mathcal{K}},\tilde{p},\mathcal{I}\right)$ be the reduced instance of $\mathcal{Q}$ as defined in Definition \ref{def:submodular-reduction}. It is easy to see that the set $\mathcal{I}$ is the independent sets of a partition matroid, as for each item $i$ at most one element $(i,D)$ can be chosen. Thus, $\mathcal{I}$ is the family of independent sets of a matroid as required. 

Next, $\mathcal{K}$ defines a tuple of MKCs, so all that is left to prove is that $\tilde{p}$ is non-negative, monotone and submodular. For each element $(i,D)\in E$ we can define a fixed non-negative value
\[ v((i,D)) = \sum_{t\in D :~ t-1\in D}g^+_{i,t} +\sum_{t\notin D :~ t-1\notin D}g^-_{i,t} \]
Thus the function $h(S)=\sum_{e\in S}v(e)$ is non-negative and modular. For $t=1,\ldots,T$ function $p_t$ is non-negative, monotone and submodular. Thus, due to Claim \ref{submodular-extension}, functions $(\tilde{p}_t)_{t=1}^T$ are also non-negative, monotone and submodular. This means that $\tilde{p}$ is non-negative, monontone and submodular as it is the sum of non-negative, monotone and submodular functions.
\end{proof}

Next, we will prove that operator $R$ defines an approximation factor preserving reduction.

\begin{lemma}
Let $\mathcal{Q}$ be an instance of Submodular $d$-GMK with time horizon $T$. For any feasible solution $(S_t,\mathcal{A}_t)_{t=1}^T$ for $\mathcal{Q}$ a feasible solution $\left(S,(\tilde{A}_{t,j})_{t\in [T],j\in [d]}\right)$ for $R(\mathcal{Q})$ exists such that $f_\mathcal{Q}\left((S_t)_{t=1}^T\right) = \tilde{f}(S)$.
\end{lemma}

\begin{proof}
Let $\mathcal{Q}=\left((\mathcal{P}_t)_{t=1}^T,g^+,g^-\right)$ be an instance of Submodular $d$-GMK, where $\mathcal{P}_{t} = \left(I,\mathcal{K}_{t},p_t\right)$ and $\mathcal{K}_t = (K_{t,j})_{j=1}^{d_t}$. Also, let $R(\mathcal{Q})=\left(E,\tilde{\mathcal{K}},\tilde{p},\mathcal{I}\right)$ be the reduced instance of $\mathcal{Q}$, where $\tilde{\mathcal{K}}=\left(\tilde{K}_{t,j }\right)_{t\in [T], j\in [d]}$
and $\tilde{K}_{t,j}=(\tilde{w},B_{t,j}, W_{t,j})$ (see Definition \ref{def:submodular-reduction}). Consider some feasible solution $(S_t,\mathcal{A}_t)_{t=1}^T$ for $\mathcal{Q}$, where $\mathcal{A}_t = (A_{t,j})_{j=1}^{d_t}$. In the following we define a solution $\left(S, \left(\tilde{A}_{t,j}\right)_{t\in [T],~j\in [d]}\right)$ for  $R(\mathcal{Q})$.
Let 
\[S= \left\{(i,D)~\middle|~i\in I, ~D=\{t\in[T]~|~i\in S_t \}\right\}\]
It can be easily verified that $S\in \mathcal{I}$. Since $\tilde{p}_t(S)=p_t(S)$ for every $t=1,\ldots,T$ and $S\subseteq E$ it holds that 
\[
\begin{aligned}
    &\tilde{p}(S) = \\ &\sum_{t=1}^T\tilde{p}_t(S) + \sum_{(i,D)\in S}\left(\sum_{t\in D :~ t-1\in D}g^+_{i,t} +\sum_{t\notin D :~ t-1\notin D}g^-_{i,t} - \sum_{t\in D:~ t-1\notin D}c^+_{i,t} - \sum_{t\in D:~ t+1\notin D}c^-_{i,t} \right) =\\
    & \sum_{t=1}^T p_t(S_t) + \sum_{t\in [2,T]}\left( \sum_{i\in S_{t-1}\cap S_{t}} g^+_{i,t} + \sum_{i\notin S_{t-1}\cup S_{t}} g^-_{i,t} \right) - \sum_{t=1}^{T}\left( \sum_{i\in S_{t}\setminus S_{t-1}} c^+_{i,t} + \sum_{i\in S_t\setminus S_{t+1}} c^-_{i,t} \right)= \\
    & f_{\mathcal{Q}}\left((S_t)_{t=1}^T \right).
\end{aligned}
\]
Next, we need to show an assignment $\tilde{A}_{t,j}$ of $S$ for each MKC in $\tilde{\mathcal{K}}$. 
 Let $t\in [T]$ and $j\in [d]$, and consider the following two cases:
 
\begin{enumerate}
    \item If $j>d_t$, recall $\tilde{K}_{t,j}=(w_0, \{b\} , W_0)$ where $w_0(i,D)=0$ for all $(i,D)\in E$ and $W_0(b)=0$ by definition. We define $\tilde{A}_{t,j}$ by $\tilde{A}_{t,j}(b)=S$. It thus holds that $w_0(\tilde{A}_{t,j}(b))=w_0(S)=0=W_0$. That is, $\tilde{A}_{t,j}$ is feasible. 
    \item  If $j\leq d_t$, then let $b^*\in B_{t,j}$ be some unique bin in $B_{t,j}$ and  we define the assignment $\tilde{A}_{t,j}:B_{t,j}\rightarrow2^E$ by 
    \begin{equation}
    \label{eq:submodular-red_assign}
    \begin{aligned}
        &\tilde{A}_{t,j}(b) = \left(A_{t,j}(b)\times 2^{[T]}\right)\cap S &\forall b\in B_{t,j}\setminus\{b^*\}\\
        &\tilde{A}_{t,j}(b^*) = \left( \left(A_{t,j}(b^*)\times 2^{[T]}\right)\cap S \right) \cup \{(i,D)\in S~|~t\not\in D\}. 
    \end{aligned}
    \end{equation}
    The assignment $\tilde{A}_{t,j}$ is a feasible assignment w.r.t $\tilde{K}_{t,j}$ since for each bin $b\in B_{t,j}$ it holds that
\[ \sum_{(i,D)\in \tilde{A}_{t,j}(b)}\tilde{w}_{t,j}((i,D)) \leq \sum_{i\in A_{t,j}(b)}w_{t,j}(i) \leq W_{t,j}(b). \]

Let $(i,D)\in S$. If $i\in S_t$ there is $b\in B_{t,j}$ such that $i\in A_{t,j}(b)$, hence $(i,D)\in \tilde{A}_{t,j}(b)$ by Equation \eqref{eq:submodular-red_assign}. If $i\not\in S_t$ then $t\notin D$ and thus $(i,D)\in \tilde{A}_{t,j}(b^*)$. Overall, we have $S\subseteq \bigcup_{b\in B_{t,j}} \tilde{A}_{t,j}(b)$. By Equation \eqref{eq:submodular-red_assign} it follows that $S\supseteq \bigcup_{b\in B_{t,j}} \tilde{A}_{t,j}(b)$ as well, thus $S= \bigcup_{b\in B_{t,j}} \tilde{A}_{t,j}(b)$. I.e, $\tilde{A}_{t,j}$ is an assignment of $S$.
\end{enumerate}
Note that the assignments can be constructed in polynomial time. We can conclude that $\left(S,\left(\tilde{A}_{t,j}\right)_{t\in[A], j\in [d]}\right)$ is a feasible solution for $R(\mathcal{Q})$, and its value is $f_{\mathcal{Q}}\left((S_t)_{t=1}^T \right)$.
\end{proof}

\begin{lemma}
Let $\mathcal{Q}$ be an instance of Submodular $d$-GMK with time horizon $T$. For any feasible solution $\left(S,(\tilde{A}_{t,j})_{t\in [T],j\in [d]}\right)$ for $R(\mathcal{Q})$ a feasible solution  $(S_t,\mathcal{A}_t)_{t=1}^T$ for $\mathcal{Q}$ can be constructed in polynomial time such that $f_\mathcal{Q}\left((S_t)_{t=1}^T\right) = \tilde{f}(S)$.
\end{lemma}

\begin{proof}
Let $\mathcal{Q}=\left((\mathcal{P}_t)_{t=1}^T,g^+,g^-\right)$ be an instance of Submodular $d$-GMK, where $\mathcal{P}_{t} = \left(I,\mathcal{K}_{t},p_t\right)$ and $\mathcal{K}_t = (K_{t,j})_{j=1}^{d_t}$. Also, let $R(\mathcal{Q})=\left(E,\tilde{\mathcal{K}},\tilde{p},\mathcal{I}\right)$ be the reduced instance of $\mathcal{Q}$. and $\left(S, \left(\tilde{A}_{t,j}\right)_{t\in [T],~j\in [d]}\right)$ be a feasible solution for $R(\mathcal{Q})$. We define the solution $(S_t,\mathcal{A}_t)_{t=1}^T$ for $\mathcal{Q}$ as follows. For every stage $t$ we set $S_t = \{i\in I~|~ \exists(i,D)\in S : t\in D\}$. For every $t=1,\ldots,T$, $j=1,\ldots,d_t$ and bin $b\in B_{t,j}$ (the set of bins in MKC $K_{t,j}$) let $A_{t,j}(b) = \{i\in I~|~ \exists (i,D)\in \tilde{A}_{t,j}(b) : t\in D\}$. Observe that the sets $(S_t)_{t=1}^{T}$ and assignments $(A_{t,j})_{t\in[T],~j\in[d_t]}$ can be constructed in polynomial time as at most $|I|$ elements can be chosen due to the matroid constraint. The assignment $A_{t,j}$ is an assignment of $S_t$ since
\[ S_t = \{i\in I~|~\exists(i,D)\in S:t\in D\} = \bigcup_{b\in B_{t,j}}\{i\in I~|~(i,D)\in \tilde{A}_{t,j}(b):t\in D\} = \bigcup_{b\in B_{t,j}}A_{t,j}(b),\]
where the second equality follows the feasibility of the solution for $R(\mathcal{Q})$. In addition, $A_{t,j}$ is a feasible assignment for MKC $K_{t,j}$ since for every bin $b\in B_{t,j}$ it holds that
\[ \sum_{i\in A_{t,j}(b)}w_{t,j}(i) = \sum_{(i,D)\in\tilde{A}_{t,j}(b) :t\in D}\tilde{w}_{t,j}((i,D)) = \sum_{(i,D)\in\tilde{A}_{t,j}(b)}\tilde{w}_{t,j}((i,D)) \leq W_{t,j}(b) \]
Thus $\left(S_t,\mathcal{A}_t\right)_{t=1}^T$ is a feasible solution for $\mathcal{Q}$.

Lastly, consider the value of the solution for $\mathcal{Q}$. It holds that
\[
\begin{aligned}
    &f_{\mathcal{Q}}\left((S_t)_{t=1}^T\right)= \\
    &\sum_{t=1}^T p_t(S_t) + \sum_{t\in [2,T]}\left( \sum_{i\in S_{t-1}\cap S_{t}} g^+_{i,t} + \sum_{i\notin S_{t-1}\cup S_{t}} g^-_{i,t}\right) - \sum_{t\in[T]}\left( \sum_{i\in S_{t}\setminus S_{t-1}} c^+_{i,t} + \sum_{i\in S_t\setminus S_{t+1}} c^-_{i,t} \right) = \\
    &\sum_{t=1}^T\tilde{p}_t(S) + \sum_{(i,D)\in S}\left(\sum_{t\in D :~ t-1\in D}g^+_{i,t} +\sum_{t\notin D :~ t-1\notin D}g^-_{i,t} - \sum_{t\in D :~ t-1\notin D}c^+_{i,t} -\sum_{t\in D :~ t+1\notin D}c^-_{i,t} \right) = \\
    & \tilde{p}(S)
\end{aligned}
\]
\end{proof}

For any Submodular $d$-GMK instance with a fixed time horizon $T$, $R(\mathcal{Q})$ can be constructed in polynomial (as $|E|=|I|\cdot 2^{|T|}$). The next corollary follows from this observation and from the last two lemmas.

\begin{corollary}
For any fixed $T\in\mathbb{N}$, there exists an approximation factor preserving reduction (see Definition \ref{def:approx-reduction}) from Submodular $d$-GMK with a time horizon bounded by $T$ to Submodular $dT$-MKCP+.
\end{corollary}

In \cite{fairstein2020tight} a $\left(1-\frac{1}{e}\right)$-approximation algorithm for Submodular $d$-MKCP+ is presented. Thus, the next lemma follows from the above corollary.
\begin{lemma}\label{lem:submodular-for-fixed-horizon}
For any fixed $T\in\mathbb{N}$ there exists a randomized $\left(1-\frac{1}{e}\right)$-approximation algorithm for $d$-GMK with a time horizon bounded by $T$.
\end{lemma}

\subsection{General Time Horizon}\label{app:submodular-cutting}

In this section we present an algorithm for Submodular $d$-GMK with a non-constant time horizon $T$. In order to do so the time horizon is cut at several stages into sub-instances. Each sub-instance is optimized independently, and then the solutions are combined to create a solution for the complete instance. A similar technique was used in \cite{bampis2019multistage}, though they considered a model without change costs. 

The following algorithm is identical to the algorithm presented in Section \ref{alg:break-time-horizon}. We present it here as well for completeness.

\begin{algorithm}[H]
	\SetAlgoLined
	\SetKwInOut{Input}{Input}
	\DontPrintSemicolon
	\Input{ A Submodular $d$-GMK instance $\mathcal{Q}$ with time horizon $T$, $0<\epsilon<\frac{1}{2}$ and $\alpha$-approximation algorithm $A$ for Submodular $d$-GMK with time horizon $T\leq\frac{2\phi_\mathcal{Q}}{\epsilon^2}$.}
	
	Set $\mu = \frac{\epsilon^2}{\phi_\mathcal{Q}}$.\;
	
	\For{$j=1,...,\frac{1}{\mu}$}{
	    Set $U_j = \left\{\frac{a}{\mu}+j-1 ~\middle|~  a\in\mathbb{N},~a\geq 1,~ \frac{a}{\mu}+j-1\leq T-\frac{1}{\mu}\right\}\cup\{0,T+1\}$.\;
	    
	    Find a solution for each cut instance in $\mathcal{Q}_{U_j}$ using algorithm $A$ and set $\mathcal{S}_j$ as the respective cut solution.\;
	}
	
	Return the solution $\mathcal{S}_j$ which maximizes the objective function $f_\mathcal{Q}$.
	
	\caption{Non Constant Time Horizon}
	\label{alg:submodular-break-time-horizon}
\end{algorithm}

\begin{lemma}\label{lem:submodular-cutting}
Let $0<\epsilon<\frac{1}{4}$, $\phi\geq 1$ and $A$ be an $\alpha$-approximation algorithm for Submodular $d$-GMK with time horizon $T\leq\frac{2\phi}{\epsilon^2}$. Also, let $\mathcal{Q}$ be an instance of Submodular $d$-GMK such that $\phi_{\mathcal{Q}}\leq \phi$. Algorithm \ref{alg:submodular-break-time-horizon} approximates $\mathcal{Q}$ within a factor of $(1-\epsilon)\alpha$.
\end{lemma}

\begin{proof}
Let $\mathcal{Q}=\left((\mathcal{P}_t)_{t=1}^T,g^+,g^-\right)$ be an instance of Submodular $d$-GMK with time horizon $T$ for which we denote its profit-cost ratio by $\phi$ and assume for simplicity it is integral. Let $0<\epsilon<\frac{1}{4}$ and $\mu = \frac{\epsilon^2}{\phi}$. Also, let $A$ be an $\alpha$-approximation algorithm for Submodular $d$-GMK with time horizon $T'\leq\frac{2\phi}{\epsilon^2}=\frac{2}{\mu}$. Note, if $T\leq\frac{2}{\mu}$, the cut points set $U_0$ will be an empty set. Thus $A$ will return an $\alpha$-approximation solution for $\mathcal{Q}$ as required. 

We show there exists a set of cut points $U_j$ and a tuple of solutions $\left((S_t,\mathcal{A}_t)_{t=u_r}^{u_{r+1}-1}\right)_{r=0}^{k-1}$ for each cut instance in $\mathcal{Q}_{U_j}=(q_j^r)_{r=0}^{k-1}$, such that the sum of values of the solutions, $\sum_{r=0}^{k-1}f_{q_j^r}\left((S_t)_{t=u_r}^{u_{r+1}-1}\right)$, is sufficiently large. From Corollary \ref{cor:cut-solution} it follows that the value of a cut solution is larger than the sum of its parts (due to lost gains and change costs saved if an item is assigned in adjacent instances). Thus, this will prove that the maximum cut solution found is sufficiently large as well. 

Let $\left(S^*_t,\mathcal{A}^*_t\right)_{t=1}^T$ be an optimal solution for $\mathcal{Q}$. Also, let $\mathcal{Q}_{U_j}=\left(q_{j,r}\right)_{r=0}^{k-1}$ be the tuple of cut instance w.r.t cut points set $U_j=\{u_0^j,\ldots,u_{k}^j\}$. For $r=0,\ldots,k-1$, the tuple $(S^*_t,\mathcal{A}^*_t)_{t=u^j_r}^{u^j_{r+1}-1}$ is a feasible solution for cut instance $q_{j,r}$ due to the feasibility of the optimal solution. The sum of values of these solutions is
\begin{align*}
\sum_{r=0}^{k-1}f_{q_j^r}\left((S^*_t)_{t=u^j_r}^{u^j_{r+1}-1}\right) &=  \sum_{r=0}^{k-1}\left(\sum_{t=u^j_r}^{u^j_{r+1}-1}p_t(S^*_t) + \sum_{t=u^j_r+1}^{u^j_{r+1}-1}\left(\sum_{i\in S_{t-1}\cap S_t}g^+_{i,t} + \sum_{i\notin S_{t-1}\cup S_t}g^-_{i,t}\right)\right) \\
&= \sum_{t=1}^Tp_t(S^*_t) + \sum_{t=[2,T]\setminus U_j}\left(\sum_{i\in S_{t-1}\cap S_t}g^+_{i,t} + \sum_{i\notin S_{t-1}\cup S_t}g^-_{i,t}\right) \\
&= f_\mathcal{Q}\left((S^*_t)_{t=1}^T\right) - \sum_{t\in U_j}\left(\sum_{i\in S_{t-1}\cap S_t}g^+_{i,t} + \sum_{i\notin S_{t-1}\cup S_t}g^-_{i,t}\right)
\end{align*}
For any set of cut points $U_1,\ldots,U_{\frac{1}{\mu}}$ such a tuple of solution can be constructed. Their total value is
\begin{align*}
\sum_{j=1}^{\frac{1}{\mu}}\sum_{r=0}^{k-1}f_{q_j^r}\left((S^*_t)_{t=u^j_r}^{u^j_{r+1}-1}\right) &= \sum_{j=1}^{\frac{1}{\mu}}\left(f_\mathcal{Q}\left((S^*_t)_{t=1}^T\right) - \sum_{t\in U_j}\left(\sum_{i\in S_{t-1}\cap S_t}g^+_{i,t} + \sum_{i\notin S_{t-1}\cup S_t}g^-_{i,t}\right)\right) \\
&\geq \frac{1}{\mu}\cdot f_\mathcal{Q}\left((S^*_t)_{t=1}^T\right) - \sum_{t=2}^T\left(\sum_{i\in S_{t-1}\cap S_t}g^+_{i,t} + \sum_{i\notin S_{t-1}\cup S_t}g^-_{i,t}\right) \\
& \geq \left(\frac{1}{\mu}-1\right)\cdot f_\mathcal{Q}\left((S^*_t)_{t=1}^T\right)
\end{align*}
There must be a value $j^*$ for which the sum of values of solutions $\sum_{r=0}^{k-1}f_{q_{j^*}^r}\left((S^*_t)_{t=u^{j^*}_r}^{u^{j^*}_{r+1}-1}\right)$ is at least as high as the average sum of values of solutions. It follows that
\begin{align*}
\sum_{r=0}^{k-1}f_{q_{j^*}^r}\left((S^*_t)_{t=u^{j^*}_r}^{u^{j^*}_{r+1}-1}\right) \geq (1-\mu)\cdot f_\mathcal{Q}\left((S^*_t)_{t=1}^T\right)
\end{align*}
Let $\left((S_t,\mathcal{A}_t)_{t=u_r}^{u_{r+1}-1}\right)_{r=0}^{k-1}$ be the solutions returned by $A$ for the cut instances. Since Algorithm $A$ returns $\alpha$-approximation solution for the problem, it holds that
\begin{align*}
\sum_{r=0}^{k-1}f_{q_{j^*}^r}\left((S_t)_{t=u^{j^*}_r}^{u^{j^*}_{r+1}-1}\right) \geq (1-\mu)\cdot\alpha\cdot f_\mathcal{Q}\left((S^*_t)_{t=1}^T\right)
\end{align*}
Finally, from Corollary \ref{cor:cut-solution} we can conclude that the value of $\mathcal{S}_{j^*}$ is at least as high as the sum of values of solution for the cut instances. Thus, its value is at least $(1-\mu)\cdot\alpha\cdot f_\mathcal{Q}\left((S^*_t)_{t=1}^T\right)$ as required.
\end{proof}

The correctness of Theorem \ref{thm:submodular-result} follows immediately from Lemma \ref{lem:submodular-cutting} and Lemma \ref{lem:submodular-for-fixed-horizon}. 

\end{document}